\documentclass[10pt,journal,letterpaper]{IEEEtran}

\usepackage{enumerate}
\usepackage{xcolor}
\usepackage{cite}
\usepackage[cmex10]{amsmath}
\usepackage{amsthm}
\usepackage{graphicx}
\graphicspath{{./figs/}}
\usepackage[caption=false,font=footnotesize]{subfig}
\usepackage{bm}
\usepackage{pifont} 
\usepackage{mathrsfs}
\usepackage{amsfonts}
\usepackage{mathabx}

\usepackage{physics}

\newcommand{\nop}[1]{}

\newtheorem{definition}{Definition}
\newtheorem{example}{Example}
\newtheorem{theorem}{Theorem}
\newtheorem{corollary}{Corollary}
\newtheorem{lemma}{Lemma}

\newtheorem{remark}{Remark}

\renewcommand{\Pr}{\mathbb{P}}

\begin{document}

\title{Performance Analysis of Quantum Channels}
\author{
Fengyou Sun\\
Department of Information Security and Communication Technology\\ 
NTNU -- Norwegian University of Science and Technology\\ 
Trondheim, Norway\\
sunfengyou@gmail.com
}

\maketitle

\begin{abstract}
We study the quality of service in quantum channels.
We regard the quantum channel as a queueing system, and present queueing analysis of both the classical information transmission and quantum information transmission in the quantum channel.
For the former, we link the analysis to the classical queueing model, for the latter, we propose a new queueing model and investigate the limit queueing behavior.
For both scenarios, we obtain tail distributions of the performance measures, i.e., backlog, delay, and throughput.
\end{abstract}


\IEEEpeerreviewmaketitle

\section{Introduction}

The era of quantum technology is coming accompanying the second quantum revolution \cite{dowling2003quantum}\cite{de2016quantum} and the future of quantum technology lies in quantum networking \cite{jason2017quantum}\cite{kimble2008quantum}\cite{castelvecchi2018entangled}.
Technically, the development of quantum internet requires a hybrid of technologies, which combine the features of both discrete variable systems and continuous variable systems \cite{pirandola2016unite}.
Theoretically, it is necessary to build a system theory for the dimension of the network dynamics, i.e., backlog, delay, and throughput, to deal with the diverse quality-of-service requirements of the network applications and to help deploy the quantum network.

In this paper, we consider three questions raised by quantum channel performance analysis and aim to provide a mathematical tool to facilitate the quantum system analysis and design.
\begin{enumerate}
\item
What is the operational utility of the quantum channel capacity?

A quantum channel can represent any physical operation that reflects the state evolution of a quantum system, ranging from an optical fiber and a free-space link to a computer memory \cite{gyongyosi2018survey}. 
There are many different kinds of quantum channel capacity depending on the involved purpose, protocol, and resource \cite{bennett2014quantum}, and these quantum channel capacity concepts are generally non-additive \cite{smith2010quantum}.
However, the information transmission follows the causal property of nature and the amount of transmitted information is additive, in addition, the performance analysis requires a study of the cumulative process the capacity process.
Therefore, we propose a new concept of quantum channel capacity, namely cumulative capacity, which is a sum of the quantum channel capacity over a period of time.
The cumulative capacity satisfies the additive-and-causal requirement in the operational domain and is explicitly used in performance analysis.

\item
What is the uniqueness of quantum channel performance analysis?

Quantum communication has many characteristics distinguishing from the classical communication, e.g., the super-activation \cite{smith2008quantum} and the negative information \cite{horodecki2005partial}.
We differentiate between the performance analysis of classical information transmission and quantum information transmission.
For classical information transmission, we show that the classical queue model is able to describe the queueing behavior in the quantum channel.
For quantum information transmission, since the quantum channel is capable of not only transmitting information but also generating communication potential due to entanglement \cite{horodecki2005partial}, we propose a new queue model, which is able to describe the fluctuation of both the communication workload and the communication potential.

\item
What is the probabilistic characterization of the information quantity in quantum channel?

The quantum channel encompasses the classical regime and the quantum regime \cite{zurek2006decoherence}, where the information adheres respectively to the classical randomness and the quantum randomness \cite{zeilinger2000quantum}\cite{zeilinger2005message}. 
Contrast to the information content, we focus on the quantity of the classical information and quantum information, i.e., the storage space of the information in the operational sense \cite{horodecki2005partial}.
We treat the quantum channel capacity as a bridge from the quantum regime to the classical regime, i.e., it takes into account the quantum effect in the quantum regime and maps onto the information transmission amount in the classical regimes.
We use classical probability to describe the randomness and regularity of the information quantity, considering the dependence that is caused by the environment in the classical regime.
\end{enumerate}

The remainder of this paper is structured as follows. 
In Sec. \ref{quantum-channel}, we recapitulate the basic concepts of quantum channel capacity, review a few explicit expressions of the classical capacity and quantum capacity of practical channels, and introduce the cumulative capacity concept. 
In Sec. \ref{performance-analysis}, we present the queueing principles for both classical information and quantum information transmission, and obtain generic results for the performance measures towards a framework for performance analysis of quantum systems.
Finally, we conclude this paper and discuss some potential research directions in Sec. \ref{conclusion}.

\section{Quantum Channel}
\label{quantum-channel}

Consider a quantum system in a Hilbert space $\mathcal{H}=\mathcal{H}_Q$, the quantum states are given by the density operator $\rho$ on $\mathcal{H}$ \cite{holevo2001evaluating}.
In the Schr\"{o}dinger picture, a quantum channel is a transformation $\rho \mapsto \mathcal{N}(\rho)$, which is a completely positive trace preserving map on trace class operators.
In this sense, the quantum channel arises from a unitary interaction $U$ between the quantum system and the environment described by another Hilbert space $\mathcal{H}_E$ with initial state $\rho_E$, i.e.,
\begin{equation}
\mathcal{N}(\rho) = \Tr_{E} \qty[ U\rho \otimes \rho_{E} U^\dagger],
\end{equation}
where $\Tr_{E}$ denotes the partial trace with respect to $\mathcal{H}_E$.
Denote $H(x)=-\Tr x\log x$ as the von Neumann entropy of a density operator $x$.
For the input state $\rho$ and the output state $\mathcal{N}(\rho)$, we have three entropy entities related to $(\rho,\mathcal{N})$, i.e., the entropy of the input state $H(\rho)$, the entropy of the output state $H(\mathcal{N}(\rho))$, and the entropy exchange $H(\rho,\mathcal{N})=H(\rho_{E}^\prime)$, where $\rho_E^\prime = \Tr_Q \qty[ U\rho \otimes \rho_{E} U^\dagger]$ is the final state of the environment.
This channel specification is extensible to the general case with different Hilbert spaces $\mathcal{H}_A$ and $\mathcal{H}_B$ of the input and output states.

\subsection{Quantum Channel Capacity}

The quantum channel $\mathcal{N}$ has many capacity concepts due to diverse transmission purpose and resource auxiliary \cite{shor2003capacities}\cite{smith2010quantum}.
The classical capacity $C(\mathcal{N})$ and quantum capacity $Q(\mathcal{N})$ quantify respectively the maximal rate of classical information and quantum information that the quantum channel can asymptotically transmit with vanishing errors. Particularly, if the transmitted classical information is secret from the environment, the resulting private classical capacity $P(\mathcal{N})$ quantifies the capability for quantum cryptography.

When prior shared entanglement between the transmitter and receiver is not available, single-letter capacity formulas are not tractable in general, and regularization is required for explicit expressions, i.e.,
\begin{IEEEeqnarray}{rCl}
\chi(\mathcal{N}) \le C(\mathcal{N}) &=& \lim_{n\rightarrow\infty} \frac{1}{n} \chi \qty({\mathcal{N}}^{\otimes n}), \\
P^{(1)} \qty({\mathcal{N}}) \le P(\mathcal{N}) &=& \lim_{n\rightarrow\infty} \frac{1}{n} P^{(1)} \qty({\mathcal{N}}^{\otimes n}), \\
Q^{(1)} \qty({\mathcal{N}}) \le Q(\mathcal{N}) &=& \lim_{n\rightarrow\infty} \frac{1}{n} Q^{(1)} \qty({\mathcal{N}}^{\otimes n}).
\end{IEEEeqnarray}
Denote the input $A$, the output $B$, the environment $E$, and the quantum mutual information $I(X;Y)=H(X)+H(Y)-H(X,Y)$.
The single-letter expression of classical capacity is \cite{smith2010quantum}
\begin{equation}
\chi(\mathcal{N}) = \max_{p_x,\rho_x} I(X;B)_\sigma,
\end{equation}
which is evaluated on state $\sigma = \sum_x p_x\ket{x}\bra{x}_X \otimes \mathcal{N}\qty(\rho_x)$.
The single-letter expression of private classical capacity is \cite{smith2010quantum}
\begin{equation}
P^{(1)} \qty({\mathcal{N}}) = \max_{p_x,\rho_x} I(X;B)_\sigma - I(X;E)_\sigma,
\end{equation}
where $\sigma = \sum_x p_x \ket{x}\bra{x}_X \otimes U \qty( \rho_x \otimes \ket{0}\bra{0}_E )U^\dagger$.
The single-letter expression of quantum capacity is \cite{smith2010quantum}
\begin{equation}
Q^{(1)} \qty({\mathcal{N}}) = \max_\rho \qty( H(B) - H(E) ),
\end{equation}
where the entropies are evaluated on the state $\sigma_{BE}=U\rho\otimes \ket{0}\bra{0}_E U^\dagger$.

When prior shared entanglement is available, the entanglement assisted capacities have single-letter formulas \cite{smith2010quantum}, i.e.,
\begin{IEEEeqnarray}{rCl}
C_E(\mathcal{N}) &=& \max_{\phi_{AA'}} I(A;B)_\sigma, \\
Q_E(\mathcal{N}) &=& \frac{1}{2} C_E(\mathcal{N}),
\end{IEEEeqnarray}
where $\sigma=I\otimes\mathcal{N}\qty(\phi_{AA'})$ and the state $\phi_{AA'}$ is unrestricted.

\subsection{Explicit Example} 
\label{application}

We review some quantum channel capacities with exact expressions.

\subsubsection{Classical Capacity of Lossy Channel}

The bosonic channels use a collection of bosonic modes to transmit information \cite{wolf2007quantum}.
The bosonic channel is a continuous-variable system and has an infinite-dimensional Hilbert space \cite{weedbrook2012gaussian}\cite{holevo2001evaluating}\cite{eisert2007gaussian}. 
The $N$ bosonic modes correspond to $N$ quantized radiation modes of the electromagnetic field. 
%
%
Consider the multi-mode bosonic channel, $\mathcal{N} = \bigotimes_{k}\mathcal{N}_{k}$, where $\mathcal{N}_k$ is the loss map for the $k$th mode, which derives from the Heisenberg evolution $a_{k}^{\prime} = \sqrt{\eta_k} a_k + \sqrt{1-\eta_k} b_k$,
where $b_k$ is the vacuum noise mode, $a_k$ and $a_k^{\prime}$ are the annihilation operators of the input and output modes, and $0\le \eta_k \le1$ is the mode transmissivity.
For the capacity to converge, the mean energy $\mathcal{E}$ of the input state is constrained.

The classical capacity of the lossy bosonic channel, in bits per channel use, is expressed as \cite{giovannetti2004classical}
\begin{equation}
C = \max_{N_k} \sum_{k} g\qty( \eta_k N_k ),
\end{equation}
where $g(x)\equiv (x+1)\log_2(x+1) - x\log_{2}x$, and the maximization is performed on the average photon number set $\qty{N_k}$ that satisfies the energy constraint
$
\sum_{k} \hbar \omega_k N_k = \mathcal{E},
$
where $\omega_k$ is the frequency of the $k$the mode.
This capacity formula applies to the lossy channel with minimum noise \cite{giovannetti2004classical}.

\begin{example}
Consider a broadband channel, where the transmitter may use all frequencies $\omega\in [0, \infty)$ and all frequencies having the same channel transmissivity $\eta$. The capacity, in bits per second, is expressed as \cite{giovannetti2004classical}\cite{shapiro2004capacity}
\begin{equation}
C = \frac{\sqrt{\eta}}{\ln 2}\sqrt{\frac{\pi P}{3\hbar}},
\end{equation}
where$P=\mathcal{E}/T$ is the average transmitted power and $T = 2\pi/\Delta\omega$ is the transmission time.
\end{example}

\begin{example}
Consider a free-space optical channel \cite{giovannetti2004classical}\cite{giovannetti2004classicalfree}, where the transmitter and the receiver communicate through circular apertures of areas $A_t$ and $A_r$, separated by a $L$ meter propagation path.
In the far field regime, only a single spatial mode in the transmitter couples a significant amount of power to the receiver. Such is the case at frequency $\omega$, the transmissivity $\eta(\omega/\omega_0)=D(\omega) \ll 1$, where $D(\omega)= (\omega/\omega_0)^2$ and $\omega_0 = 2\pi cL/\sqrt{A_t A_r}$ are respectively the Fresnel number and Fresnel frequency.
For a broadband channel with maximum transmitter frequency $\omega_c$ and $D(\omega_c) \ll 1$, the capacity, in bits per second, is expressed as \cite{giovannetti2004classical}\cite{giovannetti2004classicalfree}
\begin{equation}
C = \frac{\omega_c }{2\pi y_0} \int_{0}^{y_0} \dd{x} g\qty(\frac{1}{e^{1/x}-1}),
\end{equation} 
where $y_0$ is a dimensional less parameter that is determined by the power constraint
$P = P_0 \int_{0}^{y_0} \frac{\dd{x}}{x} \frac{1}{e^{1/x}-1}$ and $P_0 = \frac{2\pi \hbar c^2 L^2}{A_t A_r}$ is a reference power for normalization.
\end{example}

\subsubsection{Quantum Capacity of Degradable Channel}

A Gaussian channel is of form $\mathcal{N}(\rho) = \Tr_{E}\qty[ U(\rho\otimes \rho_E) U^\dagger ]$, where $U$ is a Gaussian unitary, determined by a quadratic bosonic Hamiltonian, and $\rho_E$ is a Gaussian state \cite{holevo2011probabilistic}.
A channel $\mathcal{N}(\rho) = \Tr_E[U (\rho\otimes\rho_E) U^\dagger ]$ is degradable \cite{devetak2005capacity}\cite{caruso2006degradability}, if it can be degraded to its conjugate $\mathcal{N}^c = = \Tr_B[U (\rho\otimes\rho_E) U^\dagger ]$, i.e., there is a map $\mathcal{T}: \mathcal{H}_B \mapsto \mathcal{H}_E$ such that $\mathcal{N}^c = \mathcal{T} \circ \mathcal{N}$, where $\circ$ denotes the composition of operators.
A large class of Gaussian channels are degradable \cite{devetak2005capacity}\cite{caruso2006degradability}\cite{wolf2007quantum}, e.g., the lossy channel.

The quantum capacity of the degradable Gaussian channel $\mathcal{N}=\bigotimes_{k}\mathcal{N}_k$, in qubits per channel use, is expressed as \cite{wolf2007quantum}
\begin{equation}
Q = \sum_{k} \sup_{\rho_G} J\qty(\rho_G, \mathcal{N}_k),
\end{equation}
where $J\qty(\rho_G, \mathcal{N}_k)$ is the coherent information and the supremum is taken over the Gaussian input states $\rho_{G}$.

\begin{example}
Consider a single-mode attenuation (amplification) channel with transmissivity $\eta$ (gain $\sqrt{\eta}$). The capacity, in qubits per channel use, is expressed as \cite{wolf2007quantum}
\begin{equation}
Q = \log_2|\eta| - \log_2|1-\eta|. 
\end{equation}
Let $\eta = e^{-l/l_a}$, $l$ and $l_a$ are respectively the transmission length and the absorption length for a transmission link, while $l$ and $l_a$ are respectively the storage and the decay time for quantum memories.
\end{example}

We present a discrete-variable degradable quantum channel, which is complementary to the bosinic Gaussian channel with continuous-variable quantum system and environment.

\begin{example}
Consider qubit channels with a qubit environment, both of two dimensions \cite{wolf2007quantum-qubit}.
Consider the Kraus operator representation, $\rho \mapsto \mathcal{N}(\rho) = \sum_{i=1}^{2}A_i \rho A_i^\dagger$, with $A_1 = [\cos(\alpha)\ 0;\ 0\ \cos(\beta)]$ and $A_2 = [0\ \sin(\beta);\ \sin(\alpha)\ 0]$.
The supremum of coherent information is taken over the diagonal input states.
In the region of nonzero capacity, $\cos(2\alpha)/\cos(2\beta)>0$, the capacity, in qubits per channel use, is expressed as
\begin{IEEEeqnarray}{rCl}
Q = && \max_{p\in[0,1]} h\qty( p\cos^2(\alpha) + (1-p)\sin^2(\beta) ) \nonumber\\
&& -h\qty( p\sin^2(\alpha) + (1-p)\sin^2(\beta) ),
\end{IEEEeqnarray}
where $h(x) = -x\log_{2}x - (1-x)\log_2(1-x)$ is the binary entropy function.
For $\alpha =\beta$ and $\beta=0$, it represents respectively a dephasing channel and an amplitude damping channel.
\end{example}

\subsection{Operational Extension}

We study how to use the quantum channel capacity for operational purpose
and we need a capacity concept to quantify the transmission capability of the quantum channel through a sequence of time slots. 

A direct approach is to define the quantum channel capacity through consecutive quantum channel uses.
Consider the single-letter formula $f(\mathcal{N})$ and the regularization $\overline{f}(\mathcal{N}) = \lim_{n\rightarrow\infty}\frac{1}{n} f\qty(\mathcal{N}^{\otimes n})$, where $f$ represents $\chi$, $P^{(1)}$, and $Q^{(1)}$.
The regularization implies that the capacity is always additive on parallel use of the same channel, i.e., $\overline{f}\qty(\mathcal{N}^{\otimes n}) = n\overline{f}(\mathcal{N})$. 
If $f(\mathcal{N})$ is additive, then $f\qty(\mathcal{N}^{\otimes n})=n f(\mathcal{N})$, $\overline{f}(\mathcal{N})=f(\mathcal{N})$, and $\overline{f} \qty(\mathcal{N}_{t_1}  \otimes \ldots \otimes \mathcal{N}_{t_n}) = {f} \left( \mathcal{N}_{t_1} \right) + \ldots + f\left( \mathcal{N}_{t_n} \right)$.
In general, if the additivity of $f(\mathcal{N})$ is not known, then
\begin{equation}
\overline{f} \qty(\mathcal{N}_{t_1}  \otimes \ldots \otimes \mathcal{N}_{t_n}) \ge \overline{f} \left( \mathcal{N}_{t_1} \right) + \ldots + \overline{f} \left( \mathcal{N}_{t_n} \right),
\end{equation}
which indicates that the capacity on consecutive use of different channels is super-additive.
In addition, the tensor product indicates that this definition does not take into account the dependence between the quantum channels at different time.
Moreover, this type of definition is non-causal and unrealistic, since the information can not rely on the future channel for transmission at present.

Operationally, the information transmission is additive, because the amount of information is the amount of the storage space \cite{horodecki2005partial}. 
In view of this, we propose a new capacity concept, cumulative capacity.

\begin{definition}
The sum of the capacity through a period of time $[t_1,t_n]$ is defined as cumulative capacity, i.e.,
\begin{equation}
S \qty({t_1}, {t_n}) := f^\ast \left( \mathcal{N}_{t_1} \right) + \ldots + f^\ast \left( \mathcal{N}_{t_n} \right),
\end{equation}
where $f^\ast$ represents $C$, $P$, $Q$, etc.
\end{definition}

By definition, the cumulative capacity has strict additivity property and the temporal dependence in the quantum channel is explicitly involved.

\begin{lemma}
The cumulative capacity is additive over time, i.e., for $t_m\le t_k \le t_n$,
\begin{equation}
S \qty(t_m, t_n) = S\qty(t_m,t_k) + S\qty(t_{k+1},t_n).
\end{equation}
\end{lemma}
\begin{proof}
The proof directly follows the definition.
\end{proof}

\begin{remark}
The additivity of classical and quantum information transmission is an operational reality in practice, while the non-additivity of the quantum channel capacity is a mathematical issue in quantum Shannon theory.
The regularization resolves the need for a mathematical expression of the quantum channel capacity, which is additive on the parallel use of the same channel. In general, the private capacity and quantum capacity are non-additive, and the additivity of the classical capacity is unknown \cite{smith2010quantum}, which indicates that the exact transmission capability of the quantum channel is still unknown. 
Instead, the cumulative capacity concepts defines the actual transmission amount of the quantum channel, which can be based on either the existing quantum channel capacity concepts or the postulation of an exact capacity formula.
\end{remark}

\begin{remark}
The accumulation of the asymptotic capacity through time stresses the ultimate transmission capability the channel can achieve in one time slot. Alternatively, the accumulation of the one-shot capacity stresses the finite channel uses in reality. The definition of the cumulative capacity is able to describe both accumulation scenarios and to involve the temporal dependence in capacity.
\end{remark}

\section{Performance Analysis}
\label{performance-analysis}

We provide queueing analysis of both classical and quantum information transmission through the quantum channel, with a focus on the latter. 
We investigate the statistical distribution of the performance measures of the quantum channel, and obtain both general results, which have no specifications of the arrival process and the capacity process, and specific results, which refines the general results taking advantage of the dependence property of the underlying processes.
We denote $\bigvee(X,Y)=\max(X,Y)$ and $\bigwedge(X,Y)=\min(X,Y)$.

\subsection{Queueing of Classical Information}

We show that the classical queue model is able to describe the queueing behavior of the classical information in the quantum channel with classical storage at the transmitter and receiver terminals.

\subsubsection{Queueing Principle}

The channel is essentially a classical queueing system with cumulative service process $S(t)$ and cumulative arrival process $ A(0,t)=\sum\limits_{s=0}^{t}a(s)$, 
where $a(t)$ denotes the traffic input to the channel at time slot $t$,
and the temporal increment in the system is expressed as
\begin{equation}
X(t) = a(t)-C(t).
\end{equation} 
The queueing behavior of the channel is expressed through the backlog in the system, which is a reflected process of the temporal increment $X(t)$ \cite{asmussen2003applied}, i.e.,
\begin{equation}
B(t+1) = \left[ B(t) + X(t)\right]^{+},
\end{equation}
where $[\cdot]^{+} := \bigvee(\cdot,0)$.
By iteration,
the backlog function is expressed as
\begin{equation}
B(t) = B(0) + \sup_{0\le{s}\le{t}}({A}(s,t)-{S}(s,t)). 
\end{equation}

Assume no loss, the output is the difference between the input and backlog, 
\begin{IEEEeqnarray}{rCl}
A^{\ast}(t) &=& A(t) - ( B(t) - B(0) ) \\
&=& \inf_{0\le s\le t} (A(0,s)+S(s,t)), \IEEEeqnarraynumspace
\end{IEEEeqnarray}
and the delay is defined via the input-output relationship, i.e., 
\begin{equation}
D(t) = \inf\left\{ d\ge{0}: A(t-d)\le A^\ast(t) \right\},
\end{equation}
which is the virtual delay that a hypothetical arrival has experienced on departure.

We presents the statistical tail probabilities of the performance measures in the following theorem. We assume $B(0)=0$, i.e., the queue is empty at the beginning.
We present the proof in Appendix \ref{classical-queue-theorem-general}.

\begin{theorem}\label{general-classical-theorem}
Consider classical information transmission. The tail of backlog is bounded by
\begin{IEEEeqnarray}{rCl}
\IEEEeqnarraymulticol{3}{l}{
\Pr( B(t)>x) = \Pr \qty{ \sup_{0\le{s}\le{t}}({A}(s,t)-{S}(s,t)) >x} 
} \IEEEeqnarraynumspace\\
&\le& \sum_{s=0}^{t} \mathbb{E}\qty[ e^{\theta(A(s,t)-S(s,t))} ] \cdot e^{-\theta x},
\end{IEEEeqnarray}
the tail of throughput is bounded by
\begin{IEEEeqnarray}{rCl}
\IEEEeqnarraymulticol{3}{l}{
\Pr \qty( A^{\ast}(t) >x ) = \Pr \qty{ \inf_{0\le s\le t} (A(0,s)+S(s,t)) >x } 
} \IEEEeqnarraynumspace\\
&\le& \bigwedge_{0\le s\le t} \mathbb{E}\qty[ e^{\theta(A(0,s) + S(s,t))} ] \cdot e^{-\theta x},
\end{IEEEeqnarray}
and the tail of delay is bounded by
\begin{IEEEeqnarray}{rCl}
\IEEEeqnarraymulticol{3}{l}{
\Pr ( D(t)>d) = \Pr \qty{ A(t-d) > A^\ast(t) } 
}\\
&\le& \sum\limits_{0\le s\le t} \mathbb{E} \qty[  e^{\theta \qty( A(s,t) - S(s,t) ) p} ]^{1/p} \mathbb{E}\qty[ e^{-\theta A(t-d,t) q} ]^{1/q}, \IEEEeqnarraynumspace
\end{IEEEeqnarray}
where $p$ and $q$ are positive with $1/p+1/q=1$.
\end{theorem}

\begin{remark}
Based on the union bound, $\Pr( B(t) \le x) = \Pr \{ \sup_{0\le{s}\le{t}}({A}(s,t)-{S}(s,t)) \le x \} \le \sum_{0\le{s}\le{t}}\Pr \{ ({A}(s,t)-{S}(s,t)) \le x \}$ and $\Pr ( D(t) \le d) = \Pr \{ \sup_{0\le{s}\le{t}}( A(t-d) - {A}(0,s)-{S}(s,t)) \le 0 \} \le \sum_{0\le{s}\le{t}} \Pr \{ ( A(t-d) - {A}(0,s)-{S}(s,t)) \le 0 \}$.
According to $\Pr(\bigwedge(X,Y)\le z) = \Pr(X\le z) + \Pr(Y\le z) - \Pr(X\le z, Y\le z)$,
$\Pr \qty( A^{\ast}(t) \le x ) = \Pr \{ \inf_{0\le s\le t} (A(0,s)+S(s,t)) \le x \} \le \sum_{0\le s\le t}\Pr \{  (A(0,s)+S(s,t)) \le x \}$.
It is easy to obtain upper bounds of the distributions or lower bounds of the tails of the performance measures, taking advantage of the Chernoff bound $\Pr(X \le x) \le \mathbb{E}[e^{-\theta X}]e^{\theta{x}}$, $\theta>0$ and the fact $\Pr(X>x)=1-\Pr(X\le x)$.
\end{remark}

\subsubsection{Distribution Refinement}

We consider both independence and dependence in the arrival and capacity processes, which are treated respectively as independently and identically distributed process and Markov additive process, which is introduced in Appendix \ref{markov-additive-process-description}.
We present the proofs in Appendix \ref{iid-classical-theorem-proof} and Appendix \ref{map-classical-theorem-proof}.

\begin{theorem}[I.I.D. Process]\label{iid-classical-theorem}
Consider a quantum channel with constant classical capacity $S(t)=C\cdot t$ and independently and identically distributed arrival process $a(t)\overset{\mathrm{d}}{=} a$.
The backlog and delay are bounded by
\begin{IEEEeqnarray}{rCl}
\Pr(B > x) &\le& e^{-\theta{x}}, \\
\Pr(D > d) &\le& e^{-\theta{C{d}}},
\end{IEEEeqnarray}
where $\theta$ is the root to $\kappa(\theta)=0$, $\kappa(\theta)= \log\int e^{\theta(a(t) -C)} F(dx)$.
The tail of throughput is bounded by
\begin{IEEEeqnarray}{rCl}
\Pr \qty( A^{\ast}(t) >x ) \le \bigwedge_{0\le s\le t} \mathbb{E}\qty[ e^{\theta a } ]^{s} \cdot e^{\theta (t-s)C } \cdot e^{-\theta x},
\end{IEEEeqnarray}
where $\theta>0$ is free for optimization.

Consider a quantum channel with independently and identically distributed classical capacity $C(t) \overset{\mathrm{d}}{=} C$ and constant arrival process $A(t)=\lambda\cdot t$.
The backlog and delay are bounded by
\begin{IEEEeqnarray}{rCl}
\Pr(B > x) &\le& e^{-\theta{x}}, \\
\Pr(D > d) &\le& e^{-\theta{\lambda{d}}},
\end{IEEEeqnarray}
where $\theta$ is the root to $\kappa(\theta)=0$, $\kappa(\theta)= \log\int e^{\theta(\lambda -C(t))} F(dx)$.
The tail of throughput is bounded by
\begin{IEEEeqnarray}{rCl}
\Pr \qty( A^{\ast}(t) >x ) \le \bigwedge_{0\le s\le t} \mathbb{E}\qty[ e^{\theta C } ]^{t-s} \cdot e^{\theta s \lambda } \cdot e^{-\theta x},
\end{IEEEeqnarray}
where $\theta>0$ is free for optimization.
\end{theorem}

\begin{theorem}[Markov Additive Process]\label{map-classical-theorem}
Consider a quantum channel with constant classical capacity $S(t)=C\cdot t$ and Markov additive arrival process $A(t)$. Conditional on the initial state $i=J_0\in E$ of the arrival process.
The backlog and delay are bounded by
\begin{IEEEeqnarray}{rCl}
\Pr_{i}(B>x) &\le&  \frac{h_{J_0}{(\theta)}}{\min_{j\in{E}}h_{j}{(\theta)}}e^{-\theta{x}}, \\
\Pr_{i}(D>d) &\le&  \frac{h_{J_0}{(\theta)}}{\min_{j\in{E}}h_{j}{(\theta)}}e^{-\theta{C{d}}},
\end{IEEEeqnarray}
where $\theta>0$ is the root to $\kappa(\theta)=0$, $\kappa(\theta)$ and $\bm{h}(\theta)$ are respectively the logarithm of the Perron-Frobenius eigenvalue and the corresponding right eigenvector of the kernel for the Markov additive process $A(t)-C\cdot t$.
The tail of throughput is bounded by
\begin{IEEEeqnarray}{rCl}
\Pr_{i} \qty( A^{\ast}(t) >x ) 
\le \bigwedge_{0\le s\le t} \frac{h_{J_0}(\theta)}{\min\limits_{j\in E}h_{j}(\theta)} e^{s\kappa(\theta) + \theta (t-s)C } \cdot e^{-\theta x}, \IEEEeqnarraynumspace
\end{IEEEeqnarray}
where $\theta>0$, $\kappa(\theta)$ and $\bm{h}(\theta)$ are respectively the logarithm of the Perron-Frobenius eigenvalue and the corresponding right eigenvector of the kernel for the Markov additive process $A(t)$.

Consider a quantum channel with Markov additive classical capacity and constant arrival process.
Conditional on the initial state $i=J_0\in E$ of the capacity process.
The backlog and delay are bounded by
\begin{IEEEeqnarray}{rCl}
\Pr_{i}(B>x) &\le&  \frac{h_{J_0}{(\theta)}}{\min_{j\in{E}}h_{j}{(\theta)}}e^{-\theta{x}}, \\
\Pr_{i}(D>d) &\le&  \frac{h_{J_0}{(\theta)}}{\min_{j\in{E}}h_{j}{(\theta)}}e^{-\theta{\lambda{d}}},
\end{IEEEeqnarray}
where $\theta>0$ is the root to $\kappa(-\theta)=0$, $\kappa(\theta)$ and $\bm{h}(\theta)$ are respectively the logarithm of the Perron-Frobenius eigenvalue and the corresponding right eigenvector of the kernel for the Markov additive process $S(t)-\lambda\cdot t$.
The tail of throughput is bounded by
\begin{IEEEeqnarray}{rCl}
\Pr_{i} \qty( A^{\ast}(t) >x ) 
&\le& \bigwedge_{0\le s\le t} \frac{\max\limits_{k\in E}h_{k}(\theta)}{\min\limits_{j\in E}h_{j}(\theta)} e^{(t-s)\kappa(\theta) + \theta s\lambda } \cdot e^{-\theta x}, \IEEEeqnarraynumspace
\end{IEEEeqnarray}
where $\theta>0$, $\kappa(\theta)$ and $\bm{h}(\theta)$ are respectively the logarithm of the Perron-Frobenius eigenvalue and the corresponding right eigenvector of the kernel for the Markov additive process $S(t)$.
\end{theorem}

\begin{remark}
We highlight the difference of the adjustment coefficient in the decay exponent of the tail bounds between the deterministic arrival case and the deterministic capacity case.
In addition, we refer the reader to \cite{rolski1998stochastic}\cite{asmussen2010ruin}\cite{sun2018hidden} for related results with random arrival and random capacity and for approach to obtain a complementary lower bound of the tail distribution.
Particularly, the approach to get the tail lower bound in  \cite{sun2018hidden}  is different from the $\Pr(X>x)=1-\Pr(X\le x)$ approach in the previous remark.
\end{remark}

\subsection{Queueing of Quantum Information}

Consider there are quantum storage at the transmitter and receiver terminals.
To describe the queueing behavior of the quantum states, we propose a new queue model to characterize both the accumulation of the incoming workload and the communication potential in the queue.

\subsubsection{Queueing Principle}

We use the queue to store both the incoming workload and the generated communication potential, 
specifically, the positive sign of the queue size indicates the storage of the incoming workload, while the negative sign of the queue size indicates the storage of the communication potential, i.e.,
\begin{equation}
B(t+1) = B(t) + X(t+1),
\end{equation}
where 
\begin{equation}
X(t+1) = a(t+1) - [Q(t+1)]^{+},
\end{equation}
which follows that the quantum capacity can be negative \cite{horodecki2005partial}, if $Q(t+1) \ge 0$, then $X(t+1)=a(t+1)-Q(t+1)$, and if $Q(t+1)<0$, then $X(t+1)=a(t+1)-\bigvee\{ Q(t+1),0 \}$.
By iteration, we obtain
\begin{equation}
B(t) = B(0) + \sum_{i=0}^{t}a(i) - \sum_{i=0}^{t}[Q(i)]^{+}. 
\end{equation}

Assume no loss, the cumulative output is the cumulative input minus the workload backlog, i.e.,
\begin{IEEEeqnarray}{rCl}
A^{\ast}(t) &=& A(t) - [B(t) - B(0)]^{+} \\
&=& A(t) - \left[ \sum_{i=0}^{t}a(i) - \sum_{i=0}^{t}[Q(i)]^{+} \right]^{+} \IEEEeqnarraynumspace\\
&=& \bigwedge \left\{ \sum_{i=0}^{t}[Q(i)]^{+}, A(t) \right\},
\end{IEEEeqnarray}
which indicates that the cumulative output equals the cumulative input if the backlog is negative, otherwise, it equals the cumulative capacity.
Based on the input-output relationship, we define the delay as
\begin{equation}
D(t) = \inf\left\{ d\ge{0}: A(t-d)\le A^\ast(t) \right\},
\end{equation}
which is the virtual delay that a hypothetical arrival has experienced on departure.

\begin{lemma}
The distribution of delay is expressed as
\begin{IEEEeqnarray}{rCl}
\Pr ( D(t)\le d) 
= \Pr\qty{ {  A(t-d) - Q^{+}(t) \le 0 }},
\end{IEEEeqnarray}
where $Q^{+}(t) = \sum_{i=0}^{t}[Q(i)]^{+}$.
\end{lemma}

\begin{proof}
Considering $\Pr ( D(t)\le d) = \Pr\qty{ A(t-d) \le A^\ast(t) } $ and $A(t-d)-A(t)\le0$,  if $A(t-d) - Q^{+}(t)\le 0$, then $\Pr\qty{ \bigvee \{ A(t-d) - Q^{+}(t), A(t-d) - A(t) \} > 0 } = \Pr\qty{  A(t-d) - Q^{+}(t) >0}=0$; 
if $A(t-d) - Q^{+}(t)> 0$, then $\Pr\qty{ \bigvee \{ A(t-d) - Q^{+}(t), A(t-d) - A(t) \} > 0 } = \Pr\qty{  A(t-d) - Q^{+}(t) >0} =1$. 
This completes the proof.
\end{proof}

We study the stability of the queue and investigate the impact of negative drift and positive drift of the queue increment process on the extreme behavior of the performance measures.
We define $\lim_{t\rightarrow\infty} \Pr(f(t) \le x) := \Pr(\lim_{t\rightarrow\infty} f(t) \le x)$ as the distribution of $f(t)$ at $t=\infty$.
We present the proof in Appendix \ref{quantum-limit-behavior-proof}.

\begin{theorem}\label{quantum-limit-behavior}
The probability of zero delay equals the probability of empty workload or non-empty communication potential, i.e.,
\begin{equation}
\Pr(D(t) = 0) = \Pr(B(t)\le 0).
\end{equation}
Let $\lambda$ and $Q$ be respectively the steady state mean rate of the arrival process and capacity process.
Then,
\begin{IEEEeqnarray}{rCl}
0 < \lim_{t\rightarrow\infty} \Pr \qty{ D(t)= 0 | \lambda = Q }  \le 1,
\end{IEEEeqnarray}
where the equality holds when the arrival process and capacity process are both constant.

If
\begin{IEEEeqnarray}{rCl}
\Pr\qty { \sup_{t\ge 0} \qty{ A(t) - Q^{+}(t) } \uparrow +\infty \Big| \lambda > Q } &=& 1,
\end{IEEEeqnarray}
then
\begin{IEEEeqnarray}{rCl}
\lim_{t\rightarrow\infty} \Pr \qty{ D(t)= 0 | \lambda > Q } &=& 0.
\end{IEEEeqnarray}
If
\begin{IEEEeqnarray}{rCl}
\Pr\qty { \inf_{t\ge 0} \qty{ A(t) - Q^{+}(t) } \downarrow -\infty \Big| \lambda < Q } &=& 1,
\end{IEEEeqnarray}
then
\begin{IEEEeqnarray}{rCl}
\lim_{t\rightarrow\infty} \Pr \qty{ D(t)= 0 | \lambda < Q } &=& 1.
\end{IEEEeqnarray}
\end{theorem}

\begin{remark}
It is well known that the conditions $\Pr\{ \sup_{t\ge 0} \{ A(t) - Q^{+}(t) \} \uparrow +\infty | \lambda > Q \} = 1$ and 
$\Pr\{ \inf_{t\ge 0} \{ A(t) - Q^{+}(t) \} \downarrow -\infty | \lambda < Q \} = 1$ hold for random walk \cite{rolski1998stochastic}.
This stability condition indicates that, if the communication potential of the quantum channel is stored for future use, it may cause the queue to overflow in case the potential is neither sufficiently consumed nor dropped.
\end{remark}

We study the temporal behavior and the mean value of the performance measures.
We define $\lim_{t\rightarrow\infty} \mathbb{E}[f(t)] := \mathbb{E}[\lim_{t\rightarrow\infty} f(t)]$ as the mean value of $f(t)$ at $t=\infty$. 
We present the proof in Appendix \ref{limit-behavior-stationary-proof}

\begin{theorem}\label{limit-behavior-stationary}
Consider stationary quantum capacity process $Q(t)$ with mean rate ${Q}$ and stationary quantum arrival process $a(t)$ with mean rate $\lambda$. The mean of the transient backlog at steady state is expressed as
\begin{equation}
\lim_{t\rightarrow\infty} \mathbb{E} \qty[ \frac{B(t)}{t} ] = \lambda - {Q},
\end{equation}
and the mean of the transient throughput at steady state is expressed as
\begin{equation}
\lim_{t\rightarrow\infty} \mathbb{E} \qty[ \frac{A^\ast(t)}{t} ] = \bigwedge (\lambda,{Q}).
\end{equation}
Consider continuous time, the mean of delay is expressed as
\begin{equation}
\mathbb{E}\qty[ D(t) ] = \int_{d=0}^{t} d \dd{ \Pr\qty{ A(t-d) - Q^{+}(t) \le 0} },
\end{equation}
where $Q^{+}(t) = \sum_{i=0}^{t} [Q(t)]^{+}$,
and the steady state mean is 
\begin{IEEEeqnarray}{rCl}
\lim_{t\rightarrow\infty} \mathbb{E}\qty[ D(t) | \lambda \le Q ] &=& 0, \\
\lim_{t\rightarrow\infty} \mathbb{E}\qty[ D(t) | \lambda > Q ] &\ge& 0.
\end{IEEEeqnarray}
\end{theorem}

\begin{remark}
An interesting result is that the average delay does not equal the average backlog divided by the average arrival rate for $\lambda<Q$, i.e., 
\begin{equation}
\lim_{t\rightarrow\infty} \mathbb{E} \qty[ {D(t)} ] \neq \lim_{t\rightarrow\infty} \frac { \mathbb{E} \qty[ {B(t)} ] } { \mathbb{E} \qty[ {A(t)}/{t} ] },
\end{equation}
and the equality holds when $\lambda=Q$ and both sides equal zero.
An explanation is that, the backlog can be negative for $\lambda<Q$ while the delay is non-negative.
If $\lambda=Q$ is treated as the stability condition of the queue, this result corresponds to the Little's law in the classical queue model \cite{asmussen2003applied}.
In addition, it is interesting to investigate the relationship between the mean delay and mean backlog for $\lambda>Q$.
\end{remark}

\begin{remark}
It is interesting to investigate the conditional events $\Pr(B(t)\le x | B(t)\ge 0) = \Pr(0\le B(t)\le x)/\Pr(B(t)\ge 0)$ and $\Pr(B(t)\ge -x | B(t)\le 0) = \Pr(-x\le B(t)\le 0)/\Pr(B(t)\le 0)$, $\forall x\ge0$, which represent respectively the distributions of the workload and the communication potential on their own stage, the associated mean values $\mathbb{E}[B(t) | B(t)\ge 0]$ and $\mathbb{E}[B(t) | B(t)\le 0]$, and their relationships with delay $\mathbb{E}[D(t)]$. 
\end{remark}

We presents the statistical tail probabilities of the performance measures in the following theorem. We assume $B(0)=0$, i.e., the queue is empty at the beginning.
We present the proof in Appendix \ref{quantum-queue-theorem-general}.

\begin{theorem}\label{general-quantum-theorem}
Consider quantum information transmission. 
The tail of backlog is bounded by, for $x\in\mathbb{R}$,
\begin{IEEEeqnarray}{rCl}
\IEEEeqnarraymulticol{3}{l}{
\Pr (B(t)>x) 
= \Pr\qty{ \sum_{i=0}^{t}a(i) - \sum_{i=0}^{t}[Q(i)]^{+} >x } 
}\IEEEeqnarraynumspace\\
&\le& \mathbb{E}\qty[ e^{ \theta \qty( \sum_{i=0}^{t}a(i) - \sum_{i=0}^{t}[Q(i)]^{+} ) } ] \cdot e^{-\theta x},
\end{IEEEeqnarray}
The tail of throughput is bounded by
\begin{IEEEeqnarray}{rCl}
\IEEEeqnarraymulticol{3}{l}{
\Pr \qty( A^{\ast}(t) >x ) 
= \Pr \qty{ \bigwedge \left\{ \sum_{i=0}^{t}[Q(i)]^{+}, A(t) \right\} >x } 
}\IEEEeqnarraynumspace\\
&\le& \mathbb{E}\qty[ e^{ \theta \qty( \sum_{i=0}^{t}[Q(i)]^{+} + A(t) ) } ] \cdot e^{-\theta 2x},
\end{IEEEeqnarray}
and the tail of delay is bounded by
\begin{IEEEeqnarray}{rCl}
\IEEEeqnarraymulticol{3}{l}{
\Pr ( D(t)>d) = \Pr\qty{ A(t-d) > A^\ast(t) } 
} \\
&=& \Pr\qty{ {  A(t-d) - \sum_{i=0}^{t}[Q(i)]^{+} } > 0 } \\
&\le& \mathbb{E}\qty[ e^{ \theta \qty( A(t) - \sum_{i=0}^{t}[Q(i)]^{+} ) p } ]^{1/p} \mathbb{E}\qty[e^{ -\theta A(t-d,t) q}]^{1/q}, \IEEEeqnarraynumspace
\end{IEEEeqnarray}
where $p$ and $q$ are positive with $1/p+1/q=1$.
\end{theorem}

\begin{remark}
Taking advantage of the Chernoff bound $\Pr(X \le x) \le \mathbb{E}[e^{-\theta X}]e^{\theta{x}}$, $\theta>0$ and the fact $\Pr(X>x)=1-\Pr(X\le x)$, it is easy to obtain upper bounds of the distributions or lower bounds of the tails of the performance measures $\Pr(B(t)\le x)$, $\Pr(A^\ast\le x)$, and $\Pr(D(t)\le d)$.
Specifically, according to $\Pr(\bigwedge(X,Y)\le z) = \Pr(X\le z) + \Pr(Y\le z) - \Pr(X\le z, Y\le z)$, $\Pr ( A^{\ast}(t) \le x ) = \Pr \{ \bigwedge \{ \sum_{i=0}^{t}[Q(i)]^{+}, A(t) \} \le x \} \le \Pr \{ \sum_{i=0}^{t}[Q(i)]^{+} \le x \} + \Pr \{ A(t) \le x \}$.
\end{remark}

\subsubsection{Distribution Refinement}

We consider the temporal independence in the capacity process and in the arrival process and provide the performance results as follows. 
We present the proof in Appendix \ref{iid-theorem-proof}. 

\begin{theorem}[I.I.D. Process]\label{iid-theorem}
Consider a quantum channel with independently and identically distributed arrival process $a(t)\overset{\mathrm{d}}{=}a$ and i.i.d. quantum capacity $[Q(t)]^{+}\overset{\mathrm{d}}{=}Q$.
The distribution of backlog is bounded by, for some $\theta>0$,
\begin{equation}
1 - e^{t\kappa(\theta)-\theta{x}} \le \Pr({B(t)} \le x) \le e^{t\kappa(-\theta)+\theta{x}},
\end{equation}
where $\kappa(\pm\theta) = \log\mathbb{E}\left[e^{\pm\theta({a-Q})}\right]$ is the cumulant generating function of the queue increment process.
The distribution of delay is bounded by
\begin{equation}
1 - e^{t\kappa^{Q}(-\theta) + (t-d)\kappa^{A}(\theta) } \le \Pr(D(t) \le d) \le e^{t\kappa^{Q}(\theta) + (t-d)\kappa^{A}(-\theta) },
\end{equation}
where $\theta>0$, $\kappa^{A}(\pm\theta) = \log\mathbb{E}\left[e^{\pm\theta({a})}\right]$, and $\kappa^{Q}(\pm\theta) = \log\mathbb{E}\left[e^{\pm\theta({Q})}\right]$.
The distribution of throughput is bounded by
\begin{multline}
2- e^{t\kappa^{Q}(\theta) - \theta{x}} - e^{t\kappa^{A}(\theta) - \theta{x}} 
 - e^{t\kappa^{Q}(-\theta)+\theta{x}} \times e^{t\kappa^{A}(-\theta)+\theta{x}} \\
\le \Pr \qty( A^{\ast}(t) \le x ) \le 
e^{t\kappa^{Q}(-\theta) + \theta{x}} + e^{t\kappa^{A}(-\theta) + \theta{x}} \\
 - \qty( 1 - e^{t\kappa^{Q}(\theta)-\theta{x}} ) \times \qty( 1 - e^{t\kappa^{A}(\theta)-\theta{x}}),
\end{multline}
where $\theta>0$, $\kappa^{Q}(\pm\theta)$ and $\kappa^{A}(\pm\theta)$ correspond respectively to the cumulative generating function of $Q$ and $a$.
\end{theorem}

\begin{figure*}[!t]
\centering
\subfloat[Backlog, $\lambda = Q$.]{\includegraphics[width=3.5in]{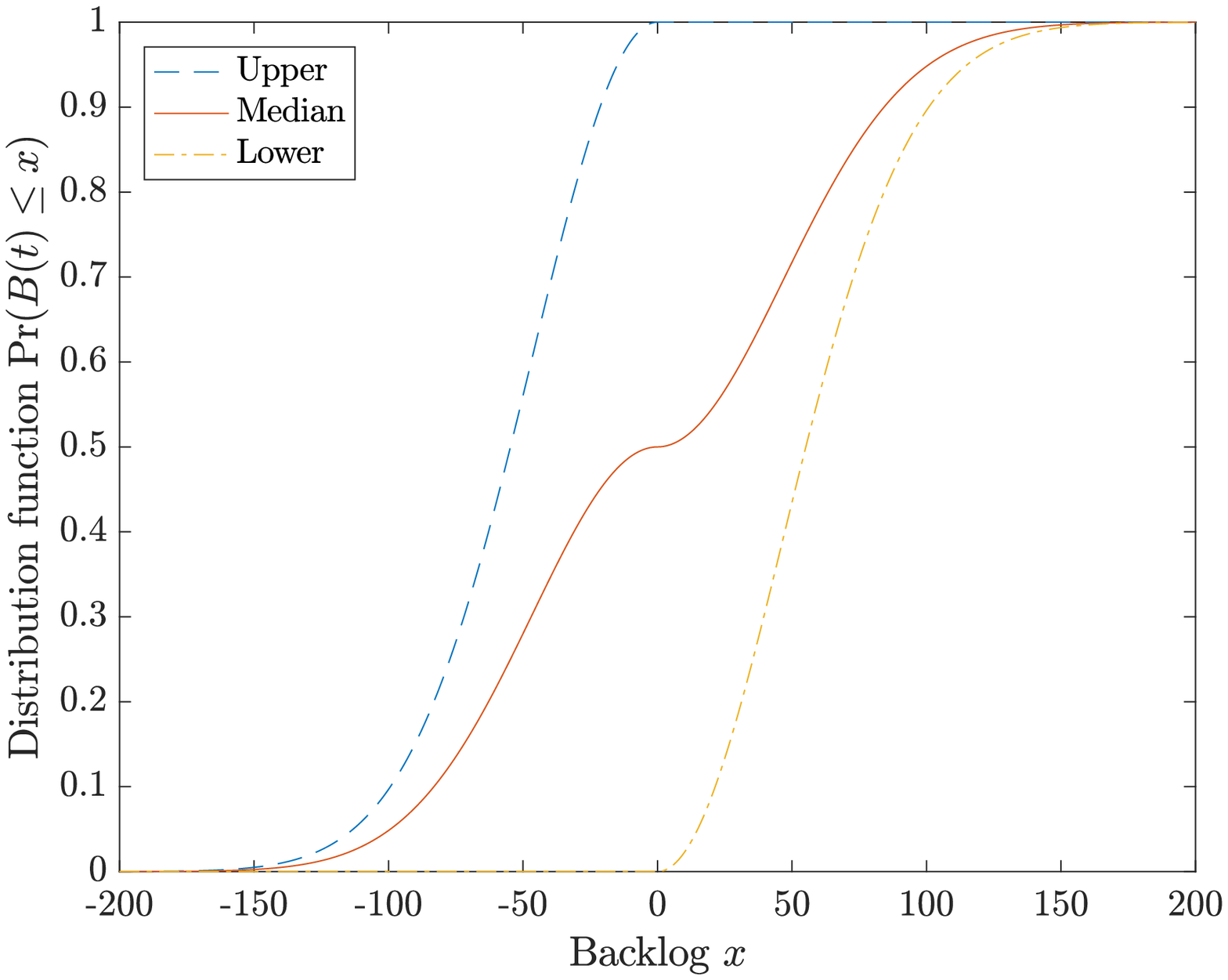}
\label{fig_first_case}}
\hfil
\subfloat[Delay, $\lambda = 10 Q$.]{\includegraphics[width=3.5in]{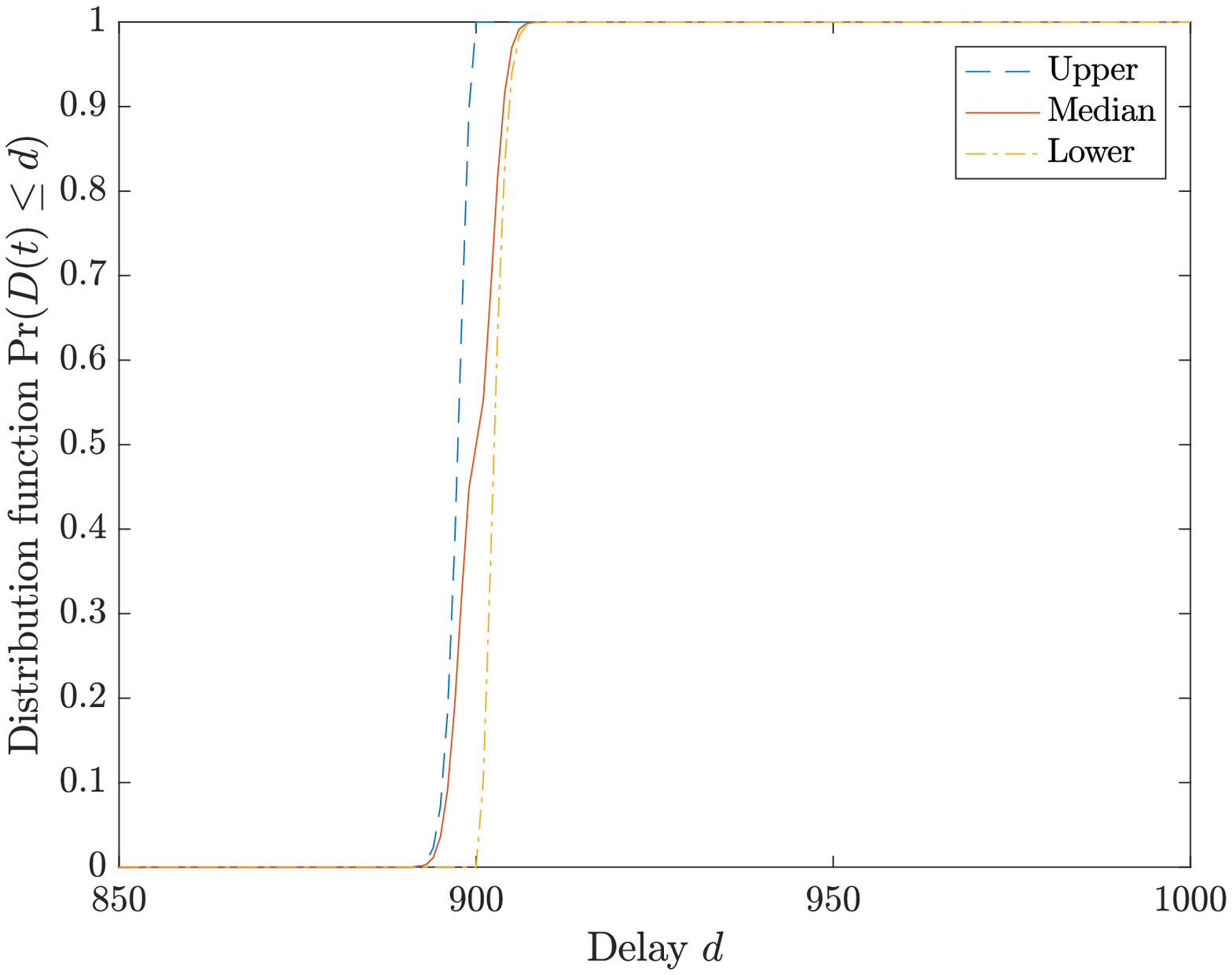}
\label{fig_second_case}}
\hfil
\subfloat[Transient throughput, $\lambda = \frac{1}{2} Q$.]{\includegraphics[width=3.5in]{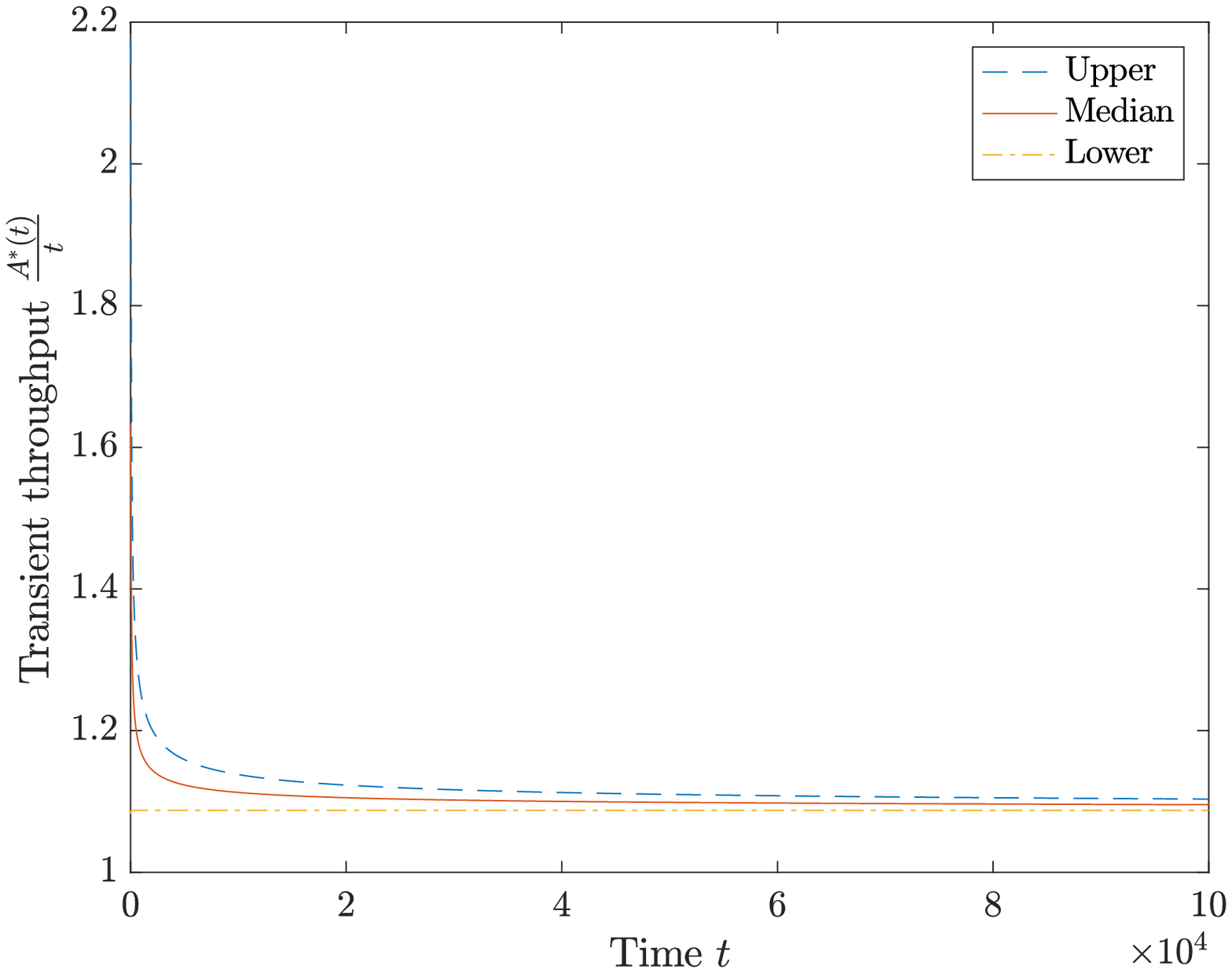}
\label{fig_third_case}}
\caption{Performance measures of quantum information transmission in bosonic Gaussian channel. Constant quantum capacity $Q=\log_{2}|\eta|-\log_{2}|1-\eta|$, where $\eta=e^{-l/l_a}$, i.i.d. quantum arrival with Poisson distribution $\mathbb{P}(n|\lambda)=\frac{\lambda^n}{n!}e^{-\lambda}$. 
$l=10$ and $l_a=50$.
$\lambda=Q$ for backlog, $\lambda=10Q$ for delay, both at time slot $t=10^3$. $\lambda=1/2 Q$ for throughput with violation probability $\Pr(A^\ast(t)/t > x)\le 10^{-5}$.
In addition to the upper and lower bounds, the median of the upper and lower bounds is illustrated.
Specifically, the upper and lower bounds imply $\mathbb{E}(D(t=10^3)) \approx 900$. 
}
\label{fig_quantum_independent}
\end{figure*}

We consider the Markov dependence, specifically, we use Markov additive process to model the cumulative arrival process and the cumulative capacity process.
We present the proof in Appendix \ref{map-theorem-proof}. 

\begin{theorem}[Markov Additive Process]\label{map-theorem}
Consider a quantum channel with Markov additive quantum capacity process $Q^{+}(t)=\sum_{i=0}^{t}[Q(t)]^{+}$ and Markov additive arrival process $A(t)$, and assume independence between the arrival process and capacity process.
Let $\kappa^A(\pm\theta)$ and $\bm{h}^A(\pm\theta)$, and $\kappa^Q(\pm\theta)$ and $\bm{h}^Q(\pm\theta)$, respectively correspond to
the logarithm of the Perron-Frobenius eigenvalue and the corresponding right eigenvector of the kernel for the Markov additive processes $A(t)$ and $Q^{+}(t)$.

Conditional on the initial states $\bm{i}=\bm{J}_0 \in \bm{E}$ of the capacity and arrival process.
The distribution of backlog is bounded by 
\begin{multline}
1- H_{-} \cdot e^{- \theta x + t\kappa^{A}(\theta) + t\kappa^{-Q}(\theta) } \le \\
\Pr_{\bm{i}}(B(t)\le x) \le 
H_{+} e^{ \theta x + t\kappa^{A}(-\theta) + t\kappa^{-Q}(-\theta) },
\end{multline}
and the distribution of delay is bounded by
\begin{multline}
1- H_{-} e^{ (t-d)\kappa^{A}(\theta) + t\kappa^{-Q}(\theta) } \le \\
\Pr_{\bm{i}}(D(t) \le d) \le H_{+} e^{ (t-d)\kappa^{A}(-\theta) + t\kappa^{-Q}(-\theta) },
\end{multline}
where $\theta>0$, $H_{-} = \frac{h_{J_0}^{A}{(\theta)}}{\min_{j\in E} h_{j}^{A}{(\theta)}} \frac{h_{J_0}^{-Q}{(\theta)}}{\min_{j\in E}h_{j}^{-Q}{(\theta)}}$ and $H_{+} = \frac{h_{J_0}^{A}{(-\theta)}}{\min_{j\in E} h_{j}^{A}{(-\theta)}} \frac{h_{J_0}^{-Q}{(-\theta)}}{\min_{j\in E}h_{j}^{-Q}{(-\theta)}}$.
The distribution of throughput is bounded by
\begin{multline}
 2 - \frac{h^{Q}_{J_0}(\theta)}{\min_{j\in E}h^{Q}_{j}(\theta)} e^{-\theta x +t\kappa^{Q}(\theta)}  
  - \frac{h^{A}_{J_0}(\theta)}{\min_{j\in E}h^{A}_{j}(\theta)} e^{-\theta x +t\kappa^{A}(\theta)}  \\
 - \frac{h^{Q}_{J_0}(-\theta) { e^{\theta x +t\kappa^{Q}(-\theta)} } }{\min_{j\in E}h^{Q}_{j}(-\theta)}  
 \times \frac{h^{A}_{J_0}(-\theta) { e^{\theta x +t\kappa^{A}(-\theta)} } }{\min_{j\in E}h^{A}_{j}(-\theta)} \\
\le \Pr_{\bm{i}}\qty(A^\ast(t) \le x)\le \\
 \frac{h^{Q}_{J_0}(-\theta)}{\min_{j\in E}h^{Q}_{j}(-\theta)} e^{\theta x +t\kappa^{Q}(-\theta)} 
 + \frac{h^{A}_{J_0}(-\theta)}{\min_{j\in E}h^{A}_{j}(-\theta)} e^{\theta x +t\kappa^{A}(-\theta)} \\
 - \qty( 1 - \frac{h^{Q}_{J_0}(\theta) { e^{-\theta x +t\kappa^{Q}(\theta)} } }{\min_{j\in E}h^{Q}_{j}(\theta)}  ) 
 \times \qty( 1 - \frac{h^{A}_{J_0}(\theta) { e^{-\theta x +t\kappa^{A}(\theta)} } }{\min_{j\in E}h^{A}_{j}(\theta)}  ).
\end{multline}
\end{theorem}

Specifically, if the capacity process or the arrival process is a constant process, we obtained the performance results as follows.
We present the proof in Appendix \ref{map-single-random-result-proof}.

\begin{corollary}[Markov Additive Process]\label{map-single-random-result}
Consider a quantum channel with constant quantum capacity $Q^{+}(t)=Q\cdot{t}$ and Markov additive arrival process $A(t)$. 
Conditional on the initial state $i=J_0 \in E$ of the arrival process.
For some $\theta>0$, the distribution of backlog is bounded by 
\begin{equation}
1 - \frac{h_{J_0}{(\theta)} e^{t\kappa(\theta)-\theta{x}} }{\min\limits_{j\in{E}}(h_{j}{(\theta)})} \le \Pr_{i}({B(t)}\le x) \le \frac{h_{J_0}{(-\theta)} e^{t\kappa(-\theta)+\theta{x}} }{\min\limits_{j\in{E}}(h_{j}{(-\theta)})},
\end{equation}
and the distribution of delay is bounded by
\begin{multline}
1- \frac{h_{J_0}(\theta)}{\min_{j\in E} h_{j}(\theta)} e^{-\theta Qt + (t-d)\kappa(\theta)} \le \\
\Pr_{i}(D(t)\le d) \le \frac{h_{J_0}(-\theta)}{\min_{j\in E} h_{j}(-\theta)} e^{\theta Qt + (t-d)\kappa(-\theta)},
\end{multline}
where $\kappa(\pm\theta)$ and $\bf{h}{(\pm\theta)}$ are respectively the logarithm of the Perron-Frobenius eigenvalue and the corresponding right eigenvector of the kernel for the Markov additive process $A(t)- Q\cdot t$ for backlog, and for $A(t)$ for delay.

Consider a quantum channel with Markov additive quantum capacity $Q^{+}(t)=\sum_{i=0}^{t}[Q(t)]^{+}$ and constant arrival process $A(t)=\lambda\cdot{t}$. 
Conditional on the initial state $i=J_0 \in E$ of the capacity process.
For some $\theta>0$, the distribution of backlog is bounded by 
\begin{equation}
1 - \frac{h_{J_0}{(\theta)} e^{t\kappa(\theta)-\theta{x}} }{\min\limits_{j\in{E}}(h_{j}{(\theta)})} \le \Pr_{i}({B(t)}\le x) \le \frac{h_{J_0}{(-\theta)} e^{t\kappa(-\theta)+\theta{x}} }{\min\limits_{j\in{E}}(h_{j}{(-\theta)})},
\end{equation}
and the distribution of delay is bounded by
\begin{multline}
1 - \frac{h_{J_0}(-\theta)}{\min_{j\in E} h_{j}(-\theta)} e^{\theta (t-d)\lambda + t\kappa(-\theta)} \le\\
\Pr_{i}(D(t) \le d) \le
\frac{h_{J_0}(\theta)}{\min_{j\in E} h_{j}(\theta)} e^{-\theta (t-d)\lambda + t\kappa(\theta)},
\end{multline}
where $\kappa(\pm\theta)$ and $\bf{h}{(\pm\theta)}$ are respectively the logarithm of the Perron-Frobenius eigenvalue and the corresponding right eigenvector of the kernel for the Markov additive process $Q^{+}(t)-\lambda\cdot t$ for backlog, and for $Q^{+}(t)$ for delay.
\end{corollary}

\begin{remark}
The theorems and corollary in this section hold in general, regardless the channel specification.
\end{remark}

\begin{remark}
The distribution upper bound and lower bound hold respectively from above the mean value and from below the mean value of the considered process, and they do not hold synchronously.
\end{remark}

\begin{remark}
For both i.i.d. case and Markov additive case, the proof indicates a few constraints on the free parameter optimization.
(1) The free parameters in the distribution upper bound and lower bound of the same process should be optimized separately. (2) For the backlog results and delay results, the free parameters in the arrival process and in the capacity process should be optimized together, because they share a common change of measure. (3) For the throughput results, the free parameters in the arrival process and the capacity process can be optimized separately, because they share different change of measure, in addition, the constraint in (1) should be taken into account when the distribution upper bound and lower bound of the same process is concerned.
As an example, we demonstrate the i.i.d. case for bosonic quantum channel in Fig. \ref{fig_quantum_independent}.
\end{remark}

\section{Conclusion}
\label{conclusion}

We develop a framework for queueing analysis of classical information and quantum information transmission in quantum channels, and study the tail distribution of the performance measures.
Particularly, we propose a new queueing model for quantum information transmission, which captures the quantum channel character that the quantum capacity is preservable as entanglement pairs for future communication, and complements the classical queueing model for classical information.
We provide both generic and specific results of performance analysis.
For the generic results, we apply the union bound, Chernoff bound, and H\"{o}lder's inequality, which are general results without constraint on the underlying stochastic process, the obtained results hold in general and apply to the complex scenarios where detailed knowledge is lacking or specific analysis is difficult to track. 
As a refinement, we take advantage of the statistical properties of the specific stochastic processes for tighter results, e.g., i.i.d. process and Markov additive process.
These specific results have no constraints on the underlying distributions of the stochastic processes and hold for a class of arrival processes and quantum channel scenarios.

We highlight that the performance analysis of quantum channels provides huge research opportunities for research.
As an outlook, we list some potential research topics as follows.
\begin{enumerate}
\item
Analysis of diverse quantum channel models, e.g., multiple access channel \cite{horodecki2007quantum}, broadcast channel \cite{yard2011quantum}, and multi-hop channel \cite{sun2017performance}.
In this paper, we have considered performance analysis in a single channel scenario. It is interesting to consider the performance analysis at network scale.

\item
Analysis of quantum channel resource trade-off.
In this paper, we have considered performance analysis based on the elemental capacity concepts. It is interesting to consider trade-off capacities \cite{bradler2010trade}\cite{bennett2014quantum} and the impact of multiplexing of classical and quantum information \cite{devetak2005capacity}. 

\item
Application to one-shot capacity \cite{holevo2001evaluating}. 
In this paper, we have considered applying the framework to the asymptotic capacity as demonstration. It is interesting to use this framework to investigate the one-shot scenario to demonstrate the impact of the constraints in reality.

\item
Application to more quantum channels in practice and investigation of the impact of different system parameters.
In this paper, we have considered the bosonic channel and the qubit  channel \cite{caruso2007qubit}\cite{wolf2007quantum-qubit}\cite{giovannetti2005information}, with a default application to quantum communication. It is interesting to apply this analysis framework to other quantum systems, e.g., quantum memory and storage. 

\item
Study on the statistical property of the quantum channel capacity and its impact on the quantum channel performance. 
The framework in this paper applies to both constant and random capacity processes. It is interesting to investigate how the statistical effects in the environment influence the capacity randomness and the channel performance, e.g., the transmissivity is a random variable due to atmospheric turbulence in the free space \cite{blake2011capacity}.
\end{enumerate}

\bibliographystyle{IEEEtran}
\bibliography{main}

\appendices

\section{Markov Additive Process}\label{markov-additive-process-description}

A Markov additive process is defined as a bivariate Markov process $\{X_t\}=\{(J_t,S(t))\}$ where $\{J_t\}$ is a Markov process with state space $E$ and the increments of $\{S(t)\}$ are governed by $\{J_t\}$ in the sense that \cite{asmussen2003applied}
\begin{equation}
\mathbb{E}[f(S({t+s})-S(t))g(J_{t+s})|\mathscr{F}_t] = \mathbb{E}_{J_t,0}[f(S(s))g(J_s)].
\end{equation}
For finite state space and discrete time, a Markov additive process is specified by the measure-valued matrix (kernel) $\mathbf{F}(dx)$ whose $ij$th element is the defective probability distribution 
\begin{equation}
F_{ij}(dx)=\mathbb{P}_{i,0}(J_1=j,Y_1\in{dx}),
\end{equation}
where $Y_t=S(t)-S(t-1)$. An alternative description is in terms of the transition matrix $\mathbf{P}=(p_{ij})_{i,j\in{E}}$ (here $p_{ij}= \mathbb{P}_i(J_1=j)$) and the probability measures
\begin{equation}
H_{ij}(dx) = \mathbb{P}(Y_1\in{dx}|J_0=i, J_1=j) = \frac{F_{ij}(dx)}{p_{ij}}.
\end{equation} 

Consider the matrix $\widehat{\textbf{F}}_t[\theta]=(\mathbb{E}_i[e^{\theta{S(t)}};J_t=j])_{i,j\in{E}}$, it is proved that \cite{asmussen2003applied}
$
\widehat{\textbf{F}}_t[\theta]=\widehat{\textbf{F}}[\theta]^t,
$
where $\widehat{\textbf{F}}[\theta]=\widehat{\textbf{F}}_1[\theta]$ is a $E\times{E}$ matrix with $ij$th element $\widehat{F}^{(ij)}[\theta]=p_{ij}\int{e^{\theta{x}}H^{(ij)}}(dx)$, and $\theta\in\Theta=\{ \theta\in\mathbb{R}:\int e^{\theta{x}}H^{(ij)} (dx) < \infty \}$. By Perron-Frobenius theory, $e^{\kappa(\theta)}$ and $\textbf{h}^{(\theta)}=(h_{i}^{(\theta)})_{i\in{E}}$ are respectively the positive real eigenvalue with maximal absolute value and the corresponding right eigenvector of $\widehat{\textbf{F}}[\theta]$, i.e., $\widehat{\textbf{F}}[\theta]\textbf{h}^{(\theta)}=e^{\kappa(\theta)}\textbf{h}^{(\theta)}$.
In addition, for the left eigenvector $\textbf{v}^{(\theta)}$, $\textbf{v}^{(\theta)}\textbf{h}^{(\theta)}=1$ and $\bm{\varpi}\textbf{h}^{(\theta)}=1$, where $\bm{\varpi}=\textbf{v}^{(0)}$ is the stationary distribution and $\textbf{h}^{(0)}=\textbf{e}$.

\section{Proof of Theorem \ref{general-classical-theorem}}\label{classical-queue-theorem-general}

The tail of backlog is bounded by
\begin{IEEEeqnarray}{rCl}
\IEEEeqnarraymulticol{3}{l}{
\Pr( B(t)>x) = \Pr \qty{ \sup_{0\le{s}\le{t}}({A}(s,t)-{S}(s,t)) >x} 
} \IEEEeqnarraynumspace\\
&\le& \sum_{s=0}^{t} \Pr \qty{ {A}(s,t)-{S}(s,t) >x} \\
&\le& \sum_{s=0}^{t} \mathbb{E}\qty[ e^{\theta(A(s,t)-S(s,t))} ] \cdot e^{-\theta x},
\end{IEEEeqnarray}
where the first inequality follows the union bound and the second inequality follows the Chernoff bound.

The tail of throughput is bounded by
\begin{IEEEeqnarray}{rCl}
\IEEEeqnarraymulticol{3}{l}{
\Pr \qty( A^{\ast}(t) >x ) = \Pr \qty{ \inf_{0\le s\le t} (A(0,s)+S(s,t)) >x } 
} \IEEEeqnarraynumspace\\
&\le& \bigwedge_{0\le s\le t} \Pr \qty{ A(0,s)+S(s,t) >x } \\
&\le& \bigwedge_{0\le s\le t} \mathbb{E}\qty[ e^{\theta(A(0,s) + S(s,t))} ] \cdot e^{-\theta x},
\end{IEEEeqnarray}
where  the second inequality follows the Chernoff bound.

The tail of delay is bounded by
\begin{IEEEeqnarray}{rCl}
\IEEEeqnarraymulticol{3}{l}{
\Pr ( D(t)>d) = \Pr \qty{ A(t-d) > A^\ast(t) } }\\
&\le& \mathbb{E} \qty[ e^{\theta \sup\limits_{0\le s\le t} \qty{ A(0,t-d)- A(0,s) - S(s,t)  } } ] \IEEEeqnarraynumspace\\
&=& \mathbb{E} \qty[ e^{\theta \sup\limits_{0\le s\le t} \qty{ A(s,t) - S(s,t) } - A(t-d,t) } ] \\
&=& \mathbb{E} \qty[ e^{\theta \sup\limits_{0\le s\le t} \qty{ A(s,t) - S(s,t) } } e^{-\theta A(t-d,t) } ] \\
&\le& \mathbb{E} \qty[ e^{\theta \sup\limits_{0\le s\le t} \qty{ A(s,t) - S(s,t) } p} ]^{1/p} \mathbb{E} \qty[ e^{-\theta A(t-d,t) q} ]^{1/q} \\
&=& \mathbb{E} \qty[\sup\limits_{0\le s\le t}  e^{\theta \qty{ A(s,t) - S(s,t) } p} ]^{1/p} \mathbb{E}\qty[ e^{-\theta A(t-d,t) q} ]^{1/q} \\
&\le& \sum\limits_{0\le s\le t} \mathbb{E} \qty[  e^{\theta \qty{ A(s,t) - S(s,t) } p} ]^{1/p} \mathbb{E}\qty[ e^{-\theta A(t-d,t) q} ]^{1/q}, \IEEEeqnarraynumspace
\end{IEEEeqnarray}
where the first inequality follows the Chernoff bound, and the second inequality follows the H\"{o}lder's inequality for positive $p$ and $q$ with $1/p+1/q=1$.

\section{Proof of Theorem \ref{iid-classical-theorem}}\label{iid-classical-theorem-proof}

For a constant service process $S(t)=C\cdot {t}$, the tail of delay is bounded by
\begin{IEEEeqnarray}{rCl}
\Pr(D > d) &=& \Pr\left\{ \sup_{t\ge{0}}({A}(t)-C\cdot t) > {C{d}} \right\} \\
&\le& e^{-\theta{C{d}}},
\end{IEEEeqnarray}
where the first equality follows the time reversibility assumption, 
the last inequality follows the Lundberg's inequality \cite{rolski1998stochastic,asmussen2010ruin}, if $\theta(>0)$ satisfies the Lundberg equation $\kappa(\theta)=0$, where
$\kappa(\theta)= \log\int e^{\theta(a(t)-C)} F(dx)$.

The tail of throughput is bounded by
\begin{IEEEeqnarray}{rCl}
\IEEEeqnarraymulticol{3}{l}{
\Pr \qty( A^{\ast}(t) >x ) 
\le \bigwedge_{0\le s\le t} \mathbb{E}\qty[ e^{\theta(A(0,s) + S(s,t))} ] \cdot e^{-\theta x}
} \IEEEeqnarraynumspace\\
&=& \bigwedge_{0\le s\le t} \mathbb{E}\qty[ e^{\theta A(0,s) } ] \mathbb{E} \qty[ e^{\theta S(s,t) } ] \cdot e^{-\theta x} \\
&=& \bigwedge_{0\le s\le t} \mathbb{E}\qty[ e^{\theta a } ]^{s} \cdot e^{\theta (t-s)C } \cdot e^{-\theta x}.
\end{IEEEeqnarray}

The proofs of other results follow analogically.

\section{Proof of Theorem \ref{map-classical-theorem}}\label{map-classical-theorem-proof}

The tail of backlog is expressed as
\begin{IEEEeqnarray}{rCl}
\Pr_{i}(B>x) &=& \Pr_{i}\qty{ \sup_{t\ge0} \qty( A(t) - C\cdot t ) > x } \\
&\le& \frac{h_{J_0}{(\theta)}}{\min_{j\in{E}}h_{j}{(\theta)}}e^{-\theta{x}}, 
\end{IEEEeqnarray}
where the first equality follows the time reversibility assumption, the last inequality follows the Lundberg's inequality, if $\theta>0$ satisfies the Lundberg equation $\kappa(\theta)=0$. $\kappa(\theta)$ and $\bm{h}{(\theta)}$ are respectively the logarithm of the Perron-Frobenius eigenvalue and the corresponding right eigenvector of the kernel for the Markov additive process $A(t)-C\cdot{t}$. 

The tail of throughput is bounded by
\begin{IEEEeqnarray}{rCl}
\IEEEeqnarraymulticol{3}{l}{
\Pr_{i} \qty( A^{\ast}(t) >x ) 
\le \bigwedge_{0\le s\le t} \mathbb{E}\qty[ e^{\theta(A(0,s) + S(s,t))} ] \cdot e^{-\theta x}
} \IEEEeqnarraynumspace\\
&=& \bigwedge_{0\le s\le t} \mathbb{E}\qty[ e^{\theta A(0,s) } ] \mathbb{E} \qty[ e^{\theta S(s,t) } ] \cdot e^{-\theta x} \\
&=& \bigwedge_{0\le s\le t} \mathbb{E}\qty[ e^{\theta A(0,s) } ]  \cdot e^{\theta (t-s)C } \cdot e^{-\theta x} \\
&\le& \bigwedge_{0\le s\le t} \frac{h_{J_0}(\theta)}{\min_{j\in E}h_{j}(\theta)} \cdot e^{s\kappa(\theta)} \cdot e^{\theta (t-s)C } \cdot e^{-\theta x},
\end{IEEEeqnarray}
where the last step follows $\mathbb{E}_i\qty[e^{\theta A(t)} h_{J_t}(\theta)]=h_{J_0}(\theta) e^{t\kappa(\theta)}$.

The proofs of other results follow analogically.

\section{Proof of Theorem \ref{quantum-limit-behavior}}
\label{quantum-limit-behavior-proof}

Note $\Pr ( D(t) \le d) = \Pr\{ {  A(t-d) - \sum_{i=0}^{t}[Q(i)]^{+} } \le 0 \}$ and $\Pr ( B(t) \le x ) = \Pr\{ {  A(t) - \sum_{i=0}^{t}[Q(i)]^{+} } \le 0 \}$, and $D(t)\ge 0$. Letting $d=x=0$ yields $\Pr(D(t) = 0) = \Pr(B(t)\le 0)$.

For $\lambda=Q$,
\begin{IEEEeqnarray}{rCl}
\IEEEeqnarraymulticol{3}{l}{
\lim_{t\rightarrow\infty} \Pr \qty{ D(t)= 0 | \lambda = Q }
= \lim_{t\rightarrow\infty} \Pr \qty{ B(t) \le 0 | \lambda = Q }. 
} \IEEEeqnarraynumspace 
\end{IEEEeqnarray}
Since $\lim_{t\rightarrow\infty} B(t)$ oscillates around zero, the probability locates in $(0,1)$.

For $\lambda>Q$,
\begin{IEEEeqnarray}{rCl}
\IEEEeqnarraymulticol{3}{l}{
\lim_{t\rightarrow\infty} \Pr \qty{ D(t)= 0 | \lambda > Q }
= \lim_{t\rightarrow\infty} \Pr \qty{ B(t) \le 0 | \lambda > Q } 
} \IEEEeqnarraynumspace\\
&\le& \Pr \qty{ \sup_{t\ge 0} B(t) \le 0 | \lambda > Q } = 0.
\end{IEEEeqnarray}
For $\lambda<Q$,
\begin{IEEEeqnarray}{rCl}
\IEEEeqnarraymulticol{3}{l}{
\lim_{t\rightarrow\infty} \Pr \qty{ D(t)= 0 | \lambda < Q }
= \lim_{t\rightarrow\infty} \Pr \qty{ B(t) \le 0 | \lambda < Q } 
} \IEEEeqnarraynumspace\\
&\ge& \Pr \qty{ \inf_{t\ge 0} B(t) \le 0 | \lambda < Q } = 1.
\end{IEEEeqnarray}

\section{Proof of Theorem \ref{limit-behavior-stationary}}
\label{limit-behavior-stationary-proof}

For the throughput, 
\begin{IEEEeqnarray}{rCl}
\IEEEeqnarraymulticol{3}{l}{
\lim_{t\rightarrow\infty} \mathbb{E} \qty[ \frac{A^\ast(t)}{t} ] = \lim_{t\rightarrow\infty} \mathbb{E} \qty[ \frac{\bigwedge(Q^{+}(t), A(t))}{t} ]
} \\
&=& \lim_{t\rightarrow\infty} \mathbb{E} \qty[ \bigwedge \qty (\frac{Q^{+}(t)}{t}, \frac{A(t)}{t} ) ]\\
&=&  \mathbb{E} \qty[ \bigwedge \qty( \lim_{t\rightarrow\infty} \qty (\frac{Q^{+}(t)}{t}, \frac{A(t)}{t} ) ) ] \\
&=& \bigwedge( \lambda, {Q} ).
\end{IEEEeqnarray}
For backlog $B(t)=A(t)- Q^{+}(t)$, the proof follows analogically.

For the delay, 
to avoid non-trivial considerations, we consider continuous time setting, where similar queueing principle expressions hold.
Particularly, the continuous time setting is the limit of the discrete time setting.
Denote $f(d,t) := \Pr\qty{ A(t-d) - Q^{+}(t) \le 0}$.
$0 \le f(d',t) \le f(d,t)\le 1$, $\forall 0\le d'\le d\le t$. 
\begin{IEEEeqnarray}{rCl}
\IEEEeqnarraymulticol{3}{l}{
\lim_{t\rightarrow\infty} \mathbb{E}\qty[ D(t) ] = \lim_{t\rightarrow\infty}  \int_{0}^{t} d \dd{ f(d,t) }
} \IEEEeqnarraynumspace\\
&=& \lim_{t\rightarrow\infty} \qty( \qty( d \cdot f(d,t) ) |_{0}^{t} -  \int_{0}^{t} { f(d,t) } \dd{d} )\\
&=& \lim_{t\rightarrow\infty} \qty( t -  \int_{0}^{t} { f(d,t) } \dd{d} ), 
\end{IEEEeqnarray}
where the second equality follows the integration by parts and the third equality follows $f(t,t)=1$.

Denote $f(d,t; \lambda=Q) := \Pr\{ A(t-d) - Q^{+}(t) \le 0 | \lambda =Q \}$. Then
\begin{IEEEeqnarray}{rCl}
\IEEEeqnarraymulticol{3}{l}{
\lim_{t\rightarrow\infty} f(0,t; \lambda=Q) = \lim_{t\rightarrow\infty} \Pr\qty{ A(t) - Q^{+}(t) \le 0 | \lambda =Q}
} \IEEEeqnarraynumspace\\
&=& \Pr\qty{ \lim_{t\rightarrow\infty} \frac{ A(t) - Q^{+}(t) }{t} \le 0 | \lambda =Q} =1.
\end{IEEEeqnarray}
Thus,
\begin{IEEEeqnarray}{rCl}
\IEEEeqnarraymulticol{3}{l}{
\lim_{t\rightarrow\infty} \mathbb{E}\qty[ D(t) | \lambda = Q ] = \lim_{t\rightarrow\infty} \qty( t -  \int_{0}^{t} { f(d,t; \lambda = Q) } \dd{d} )
} \IEEEeqnarraynumspace\\
&=& 0,
\end{IEEEeqnarray}
where the second equality follows $f(0,t; \lambda = Q)\le f(d,t; \lambda = Q) \le 1$, $\forall 0\le d\le t$.

Since $\lim_{t\rightarrow\infty} f(0,t; \lambda<Q) = 1$, we obtain $\lim_{t\rightarrow\infty} \mathbb{E}[ D(t) | \lambda < Q ] = 0$.

Denote $f(d,t; \lambda>Q) := \Pr\{ A(t-d) - Q^{+}(t) \le 0 | \lambda >Q \}$. Then
\begin{IEEEeqnarray}{rCl}
\IEEEeqnarraymulticol{3}{l}{
\lim_{t\rightarrow\infty} f(0,t; \lambda>Q) = \lim_{t\rightarrow\infty} \Pr\qty{ A(t) - Q^{+}(t) \le 0 | \lambda > Q}
} \IEEEeqnarraynumspace\\
&=& \Pr\qty{ \lim_{t\rightarrow\infty} \frac{ A(t) - Q^{+}(t) }{t} \le 0 | \lambda > Q} = 0.
\end{IEEEeqnarray}
Thus,
\begin{equation}
\lim_{t\rightarrow\infty} \mathbb{E}\qty[ D(t) | \lambda > Q ] = \lim_{t\rightarrow\infty} \qty( t -  \int_{0}^{t} { f(d,t; \lambda>Q) } \dd{d} )  \ge 0, 
\end{equation}
where the second equality follows $f(0,t; \lambda > Q)\le f(d,t; \lambda > Q) \le 1$, $\forall 0\le d\le t$.

\section{Proof of Theorem \ref{general-quantum-theorem}}\label{quantum-queue-theorem-general}

The tail of the queue backlog is expressed as
\begin{IEEEeqnarray}{rCl}
\IEEEeqnarraymulticol{3}{l}{
\Pr (B(t)>x) 
= \Pr\qty{ \sum_{i=0}^{t}a(i) - \sum_{i=0}^{t}[Q(i)]^{+} >x } 
}\IEEEeqnarraynumspace\\
&\le& \mathbb{E}\qty[ e^{ \theta \qty( \sum_{i=0}^{t}a(i) - \sum_{i=0}^{t}[Q(i)]^{+} ) } ] \cdot e^{-\theta x},
\end{IEEEeqnarray}
where the inequality follows the Chernoff bound.

The tail of throughput is expressed as
\begin{IEEEeqnarray}{rCl}
\IEEEeqnarraymulticol{3}{l}{
\Pr \qty( A^{\ast}(t) >x ) 
= \Pr \qty{ \bigwedge \left\{ \sum_{i=0}^{t}[Q(i)]^{+}, A(t) \right\} >x } 
}\IEEEeqnarraynumspace\\
&\le& \Pr \qty{ \frac{ \sum_{i=0}^{t}[Q(i)]^{+} + A(t) }{2} > x } \\
&\le& \mathbb{E}\qty[ e^{ \theta \qty( \sum_{i=0}^{t}[Q(i)]^{+} + A(t) ) } ] \cdot e^{-\theta 2x},
\end{IEEEeqnarray}
where the first inequality follows $\bigwedge(X,Y)\le(X+Y)/2$, and the second inequality follows the Chernoff bound.

The tail of delay is expressed as 
\begin{IEEEeqnarray}{rCl}
\IEEEeqnarraymulticol{3}{l}{
\Pr ( D(t)>d) = \Pr\qty{ A(t-d) > A^\ast(t) } 
} \\
&=& \Pr\qty{ {  A(t-d) - Q^{+}(t) } > 0 } \\
&\le& \mathbb{E}\qty[ e^{ \theta \qty( A(t-d) - \sum_{i=0}^{t}[Q(i)]^{+} ) } ] \\
&=& \mathbb{E}\qty[ e^{ \theta \qty(A(t) - \sum_{i=0}^{t}[Q(i)]^{+} ) } e^{ -\theta A(t-d,t) } ] \\
&\le& \mathbb{E}\qty[ e^{ \theta \qty( A(t) - \sum_{i=0}^{t}[Q(i)]^{+} ) p } ]^{1/p} \mathbb{E}\qty[e^{ -\theta A(t-d,t) q}]^{1/q}, \IEEEeqnarraynumspace
\end{IEEEeqnarray}
where $Q^{+}(t) = \sum_{i=0}^{t}[Q(i)]^{+}$, 
the second inequality follows the Chernoff bound, and the third inequality follows the H\"{o}lder's inequality for positive $p$ and $q$ with $1/p+1/q=1$.

\section{Proof of Theorem \ref{iid-theorem}}\label{iid-theorem-proof}

Consider the quantum capacity $Q(t) \overset{\mathrm{d}}{=}Q$ and instantaneous arrival $a(t)\overset{\mathrm{d}}{=}a$.
A likelihood ratio process of the backlog is formulated and expressed as \cite{asmussen2003applied}
\begin{equation}
L(t) = e^{\theta{B(t)}-t\kappa(\theta)},
\end{equation}
where $L(t)$ is a mean-one martingale and $\kappa(\theta)$ is the cumulant generating function, i.e., 
$
\kappa(\theta) = \log\mathbb{E}\left[e^{\theta({a-Q})}\right] = \log\int e^{\theta{x}} F(dx),
$
where $\theta\in\Theta=\{ \theta\in\mathbb{R}:\kappa(\theta)<\infty \}$.

According to Markov inequality, for any $\mu>0$,
\begin{equation}
\Pr\{ L(t)\ge{\mu} \} \le\frac{1}{\mu}\mathbb{E}[L(t)]=\frac{1}{\mu}.
\end{equation}
Letting $\mu = e^{-t\kappa(\theta)+\theta{x}}$,
for $\theta\le{0}$, the cumulative distribution function is bounded by
\begin{equation}
\Pr\{ B(t)\le x \} \le e^{t\kappa(\theta)-\theta{x}},
\end{equation}
while for $\theta>0$, the complementary cumulative distribution function is expressed as
\begin{equation}
\Pr\{ B(t)\ge x \} \le e^{t\kappa(\theta)-\theta{x}},
\end{equation}
which shows that the distribution has a light tail.

The tail of delay is expressed as
\begin{IEEEeqnarray}{rCl}
\IEEEeqnarraymulticol{3}{l}{
\Pr ( D(t)>d) = \Pr\qty{ A(t-d) > A^\ast(t) } 
} \\
&=& \Pr\qty{ {A(t-d) - \sum_{i=0}^{t}[Q(i)]^{+} } > 0 } \\
&\le& \mathbb{E}_{\theta} \qty[ e^{ -\theta \hat{D}(t,d) + t\kappa^{Q}(-\theta) + (t-d) \kappa^{A}(\theta) }; \hat{D}(t,d) >0 ] \IEEEeqnarraynumspace\\
&\le& e^{t\kappa^{Q}(-\theta) + (t-d)\kappa^{A}(\theta) },
\end{IEEEeqnarray}
where $\hat{D}(t,d) := A(t-d)-Q^{+}(t)$.
Similarly, the tail of delay is lower bounded by
\begin{IEEEeqnarray}{rCl}
\IEEEeqnarraymulticol{3}{l}{
\Pr ( D(t)\le d) = \Pr\qty{ {A(t-d) - \sum_{i=0}^{t}[Q(i)]^{+} } \le 0 }
} \\
&\le& \mathbb{E}_{-\theta} \qty[ e^{ \theta \hat{D}(t,d) + t\kappa^{Q}(\theta) + (t-d) \kappa^{A}(-\theta) }; \hat{D}(t,d) \le 0 ] \IEEEeqnarraynumspace\\
&\le& e^{t\kappa^{Q}(\theta) + (t-d)\kappa^{A}(-\theta) },
\end{IEEEeqnarray}
where $\hat{D}(t,d) := A(t-d)-Q^{+}(t)$.

The distribution of throughput is expressed as, for $\theta>0$,
\begin{IEEEeqnarray}{rCl}
\IEEEeqnarraymulticol{3}{l}{
\Pr \qty( A^{\ast}(t) \le x ) 
= \Pr \qty{ \bigwedge \left\{ Q^{+}(t), A(t) \right\} \le x } 
}\IEEEeqnarraynumspace\\
&=& \Pr( Q^{+}(t) \le x ) + \Pr(A(t) \le x) \nonumber\\
&& - \Pr( Q^{+}(t) \le x ) \Pr(A(t) \le x) \\
&\le& e^{t\kappa^{Q}(-\theta) + \theta{x}} + e^{t\kappa^{A}(-\theta) + \theta{x}} \\
&& - \qty( 1 - e^{t\kappa^{Q}(\theta)-\theta{x}} ) \times \qty( 1 - e^{t\kappa^{A}(\theta)-\theta{x}} ),
\end{IEEEeqnarray}
similarly,
the distribution is lower bounded by
\begin{IEEEeqnarray}{rCl}
\Pr \qty( A^{\ast}(t) \le x ) 
&\ge& 2- e^{t\kappa^{Q}(\theta) - \theta{x}} - e^{t\kappa^{A}(\theta) - \theta{x}} \\
&& - e^{t\kappa^{Q}(-\theta)+\theta{x}} \times e^{t\kappa^{A}(-\theta)+\theta{x}}, \IEEEeqnarraynumspace
\end{IEEEeqnarray}
where $\kappa^{Q}(\pm\theta)$ and $\kappa^{A}(\pm\theta)$ correspond respectively to the cumulative generating function of $Q$ and $a$.

\section{Proof of Theorem \ref{map-theorem}}\label{map-theorem-proof}

We prove the results based on the change of measure approach \cite{asmussen2003applied}, represent the distribution in the changed measure, and find upper or lower bounds of the distribution.

The likelihood ratio martingale of the arrival process $A(t)$ is expressed as
\begin{equation}
L^{A}_{t} = \frac{h_{J_t}^{A}{(\theta)}}{h_{J_0}^{A}{(\theta)}} e^{\theta A(t) - t\kappa^{A}(\theta)},
\end{equation}
which is a mean-one martingale,
and the likelihood ratio mean-one martingale of the service process $-Q^{+}(t)$ is 
\begin{equation}
L^{-Q}_{t} = \frac{h_{J_t}^{-Q}{(\theta)}}{h_{J_0}^{-Q}{(\theta)}} e^{-\theta Q(0,t) - t\kappa^{-Q}(\theta)}.
\end{equation}
Assume the arrival process and the service process are independent, then the product of the martingales 
\begin{equation}
L^{A-Q}_{0,t} =  L^{A}_{t}  \cdot L^{-Q}_{t}
\end{equation}
is also a martingale \cite{cherny2006some}, and
\begin{equation}
\mathbb{E} \left[ L^{A-Q}_{0,t} \right] = \mathbb{E}\left[ L^{A}_{t}  \right] \cdot \mathbb{E}\left[ L^{-Q}_{t} \right] = 1.
\end{equation}

The tail of backlog is expressed as
\begin{IEEEeqnarray}{rCl}
\IEEEeqnarraymulticol{3}{l}{
\Pr_{\bm i} (B(t)>x) 
= \Pr_{\bm i}\qty{ A(t) - Q^{+}(t) >x } 
}\IEEEeqnarraynumspace\\
&\le& \mathbb{E}_{\theta, \bm i} \qty[ H(\theta)  e^{- \theta (A(t) - Q(t)) + t\kappa^{A}(\theta) + t\kappa^{-Q}(\theta) }; B(t)>x ] \IEEEeqnarraynumspace\\
&\le& H_{-} \cdot e^{- \theta x + t\kappa^{A}(\theta) + t\kappa^{-Q}(\theta) },
\end{IEEEeqnarray}
and 
\begin{IEEEeqnarray}{rCl}
\IEEEeqnarraymulticol{3}{l}{
\Pr_{\bm i} (B(t) \le x) 
= \Pr_{\bm i}\qty{ A(t) - Q^{+}(t) \le x } 
}\IEEEeqnarraynumspace\\
&=& \mathbb{E}_{-\theta, \bm i} \qty[ H(-\theta)  e^{ \theta B(t) + t\kappa^{A}(-\theta) + t\kappa^{-Q}(-\theta) }; B(t) \le x ] \IEEEeqnarraynumspace\\
&\le& H^{+} e^{ \theta x + t\kappa^{A}(-\theta) + t\kappa^{-Q}(-\theta) },
\end{IEEEeqnarray}
where $H_{-} = \frac{h_{J_0}^{A}{(\theta)}}{\min_{j\in E} h_{j}^{A}{(\theta)}} \frac{h_{J_0}^{-Q}{(\theta)}}{\min_{j\in E}h_{j}^{-Q}{(\theta)}}$, $H_{+} = \frac{h_{J_0}^{A}{(-\theta)}}{\min_{j\in E} h_{j}^{A}{(-\theta)}} \frac{h_{J_0}^{-Q}{(-\theta)}}{\min_{j\in E}h_{j}^{-Q}{(-\theta)}}$, and $H(\theta)= \frac{h_{J_0}^{A}{(\theta)}}{h_{J_t}^{A}{(\theta)}} \frac{h_{J_0}^{-Q}{(\theta)}}{h_{J_t}^{-Q}{(\theta)}}$.

The tail of delay is expressed as
\begin{IEEEeqnarray}{rCl}
\IEEEeqnarraymulticol{3}{l}{
\Pr_{\bm i} ( D(t)>d) = \Pr_{\bm i}\qty{ { A(t-d) -Q^{+}(t) > 0 } }
} \\
&\le& \mathbb{E}_{\theta,\bm i} \qty[ H(\theta)  e^{ -\theta \hat{D}(t) + (t-d)\kappa^{A}(\theta) + t\kappa^{-Q}(\theta) }; \hat{D}(t) > 0 ] \IEEEeqnarraynumspace\\
&\le& H_{-} e^{ (t-d)\kappa^{A}(\theta) + t\kappa^{-Q}(\theta) },
\end{IEEEeqnarray}
and
\begin{IEEEeqnarray}{rCl}
\IEEEeqnarraymulticol{3}{l}{
\Pr_{\bm i} ( D(t)\le d) = \Pr_{\bm i}\qty{ { A(t-d) -Q^{+}(t) \le 0 } }
} \\
&\le& \mathbb{E}_{-\theta,\bm i} \qty[ H(-\theta)  e^{ \theta \hat{D}(t) + (t-d)\kappa^{A}(-\theta) + t\kappa^{-Q}(-\theta) }; \hat{D}(t) \le 0 ] \nonumber\\*
&\le& H_{+} e^{ (t-d)\kappa^{A}(-\theta) + t\kappa^{-Q}(-\theta) },
\end{IEEEeqnarray}
where $\hat{D}(t) = A(t-d) - Q^{+}(t)$, $H(\theta)= \frac{h_{J_0}^{A}{(\theta)}}{h_{J_{t-d}}^{A}{(\theta)}} \frac{h_{J_0}^{-Q}{(\theta)}}{h_{J_t}^{-Q}{(\theta)}}$, $H_{-} = \frac{h_{J_0}^{A}{(\theta)}}{\min_{j\in E} h_{j}^{A}{(\theta)}} \frac{h_{J_0}^{-Q}{(\theta)}}{\min_{j\in E}h_{j}^{-Q}{(\theta)}}$, and $H_{+} = \frac{h_{J_0}^{A}{(-\theta)}}{\min_{j\in E} h_{j}^{A}{(-\theta)}} \frac{h_{J_0}^{-Q}{(-\theta)}}{\min_{j\in E}h_{j}^{-Q}{(-\theta)}}$.

The distribution of throughput is expressed as
\begin{IEEEeqnarray}{rCl}
\IEEEeqnarraymulticol{3}{l}{
\Pr_{\bm i} \qty( A^{\ast}(t) \le x ) 
= \Pr_{\bm i} \qty{ \bigwedge \left\{ Q^{+}(t), A(t) \right\} \le x } 
} \\
&=& \Pr_{\bm i}( Q^{+}(t) \le x ) + \Pr_{\bm i}(A(t) \le x) \nonumber\\
&& - \Pr_{\bm i}( Q^{+}(t) \le x ) \Pr_{\bm i}(A(t) \le x) \\
&\le& \frac{h^{Q}_{J_0}(-\theta)}{\min\limits_{j\in E}h^{Q}_{j}(-\theta)} e^{\theta x +t\kappa^{Q}(-\theta)} 
 + \frac{h^{A}_{J_0}(-\theta)}{\min\limits_{j\in E}h^{A}_{j}(-\theta)} e^{\theta x +t\kappa^{A}(-\theta)} \nonumber\\
&-&  \qty( 1 - \frac{h^{Q}_{J_0}(\theta) { e^{-\theta x +t\kappa^{Q}(\theta)} } }{\min_{j\in E}h^{Q}_{j}(\theta)}  ) 
 \times \qty( 1 - \frac{h^{A}_{J_0}(\theta) { e^{-\theta x +t\kappa^{A}(\theta)} }}{\min_{j\in E}h^{A}_{j}(\theta)}  ). \nonumber\\*
\end{IEEEeqnarray}
Similarly, the throughput distribution is lower bounded by
\begin{IEEEeqnarray}{rCl}
\IEEEeqnarraymulticol{3}{l}{
\Pr_{\bm i} \qty( A^{\ast}(t) \le x ) 
} \nonumber\\
&\ge&  2 - \frac{h^{Q}_{J_0}(\theta)}{\min\limits_{j\in E}h^{Q}_{j}(\theta)} e^{-\theta x +t\kappa^{Q}(\theta)}  
  - \frac{h^{A}_{J_0}(\theta)}{\min\limits_{j\in E}h^{A}_{j}(\theta)} e^{-\theta x +t\kappa^{A}(\theta)}  \nonumber\\
&& - \frac{h^{Q}_{J_0}(-\theta) { e^{\theta x +t\kappa^{Q}(-\theta)} } }{\min_{j\in E}h^{Q}_{j}(-\theta)}  
 \times \frac{h^{A}_{J_0}(-\theta) { e^{\theta x +t\kappa^{A}(-\theta)} } }{\min_{j\in E}h^{A}_{j}(-\theta)}. 
\end{IEEEeqnarray}

\section{Proof of Corollary \ref{map-single-random-result}} \label{map-single-random-result-proof}

For the deterministic capacity case, the backlog process forms a Markov additive process.
We define the likelyhood ratio martingale as \cite{asmussen2003applied}, $\forall \theta\in \mathbb{R}$,
\begin{equation}
{L}(t) = \frac{h_{J_t}{(\theta)}}{h_{J_0}{(\theta)}}e^{\theta{B(t)}-t\kappa(\theta)}.
\end{equation}
Then, for $\theta>0$,
\begin{IEEEeqnarray}{rCl}
\Pr_{i} ( B(t) > x ) &=& \mathbb{E}_{\theta,i} \left[ \frac{h_{J_0}(\theta)}{h_{J_t}(\theta)} e^{-\theta B(t) + t\kappa(\theta)}; B(t)>x \right] \IEEEeqnarraynumspace\\
&\le& \frac{h_{J_0}(\theta)}{\min_{j\in E}h_{j}(\theta)} e^{-\theta x +t\kappa(\theta)},
\end{IEEEeqnarray}
and, for $\theta>0$, 
\begin{IEEEeqnarray}{rCl}
\Pr_{i} ( B(t) \le x ) &=& \mathbb{E}_{-\theta,i} \left[ \frac{h_{J_0}(-\theta)}{h_{J_t}(-\theta)} e^{\theta B(t) + t\kappa(-\theta)}; B(t) \le x \right] \nonumber\\*
&\le& \frac{h_{J_0}(-\theta)}{\min_{j\in E}h_{j}(-\theta)} e^{\theta x +t\kappa(-\theta)}.
\end{IEEEeqnarray}
The tail of delay is expressed as
\begin{IEEEeqnarray}{rCl}
\IEEEeqnarraymulticol{3}{l}{
\Pr_{i} ( D(t)>d) = \Pr_{i}\qty{ A(t-d) > \sum_{i=0}^{t}[Q(i)]^{+} }
} \\
&\le& \mathbb{E}_{\theta,i} \left[ H(\theta) e^{-\theta A(t-d) + (t-d)\kappa(\theta)}; A(t-d) > Q t \right] \IEEEeqnarraynumspace\\
&\le& \frac{h_{J_0}(\theta)}{\min_{j\in E} h_{j}(\theta)} e^{-\theta Qt + (t-d)\kappa(\theta)},
\end{IEEEeqnarray}
where $H(\theta) =\frac{h_{J_0}(\theta)}{h_{J_{t-d}}(\theta)}$.
Similarly, the distribution is bounded by
\begin{IEEEeqnarray}{rCl}
\IEEEeqnarraymulticol{3}{l}{
\Pr_{i} ( D(t) \le d) = \Pr_{i}\qty{ A(t-d) \le \sum_{i=0}^{t}[Q(i)]^{+} }
} \\
&\le& \mathbb{E}_{-\theta,i} \left[ H(-\theta) e^{\theta A(t-d) + (t-d)\kappa(-\theta)}; A(t-d) \le Q t \right] \IEEEeqnarraynumspace\\
&\le& \frac{h_{J_0}(-\theta)}{\min_{j\in E} h_{j}(-\theta)} e^{\theta Qt + (t-d)\kappa(-\theta)},
\end{IEEEeqnarray}
where $H(-\theta) =\frac{h_{J_0}(-\theta)}{h_{J_{t-d}}(-\theta)}$.

The proof of other results, for the random capacity case, follows analogically.

\end{document}